\documentclass[12pt,a4paper]{scrartcl}

\usepackage[english]{babel}
\usepackage[utf8]{inputenc} 	
\usepackage{lmodern}
\usepackage{amssymb}
\usepackage{amsmath}
\usepackage{graphicx}
\usepackage{amsthm}
\usepackage{comment}
\usepackage{mathtools}
\usepackage[shortlabels]{enumitem}
\usepackage{xcolor}
\usepackage{aliascnt}
\numberwithin{equation}{section}

\usepackage{csquotes}
\usepackage{dsfont}
\usepackage[style=alphabetic,maxnames=100,giveninits=true,backend=biber,doi=false,url=false,isbn=false,eprint=true]{biblatex}
\usepackage{hyperref}
\usepackage{cleveref}

\newcommand{\inn}[1]{\langle {#1} \rangle }
\newcommand{\R}{{\mathbb{R}}}
\newcommand{\N}{{\mathbb{N}}}
\newcommand{\Z}{{\mathbb{Z}}}
\newcommand{\C}{{\mathbb{C}}}
\newcommand{\IR}{{\mathbb{R}}}
\newcommand{\IN}{{\mathbb{N}}}

\newcommand{\abs}[1]{\left| #1 \right|}

\newcommand{\muu}[1]{\begin{align*} #1 \end{align*}}
\newcommand{\muun}[1]{\begin{align} #1 \end{align}}

\newcommand{\abc}{\begin{enumerate}[label=\alph*)]}
\renewcommand{\epsilon}{{\varepsilon}}
\renewcommand{\sc}[1]{\left< #1 \right>}
\newcommand{\nn}[1]{\left\Vert #1 \right\Vert}
\newcommand{\ton}{\overset{n\longrightarrow\infty}{\longrightarrow}}

\newcommand{\ran}{\operatorname{ran}}

\newcommand{\Id}{\mathds{1}}
\newcommand{\one}{\mathds{1}}

\newcommand{\dG}{\mathsf{d}\Gamma}

\renewcommand{\i}{\mathrm{i}}

\newcommand{\HH}{\mathcal{H}}

\renewcommand{\:}{\colon}

\newcommand{\f}{\mathrm{f}}

\renewcommand{\:}{\colon}

\renewcommand{\Im}{\operatorname{Im}}
\renewcommand{\Re}{\operatorname{Re}}
\newcommand{\Hf}{H_\f}
\newcommand{\Hfm}{H_{\f,m}}

\newcommand{\Pf}{P_\f}
\newcommand{\Pfj}[1]{P_{\f,#1}}
\newcommand{\ps}{1,2}

\newcommand{\hs}{\mathfrak h}

\newcommand{\fin}{\operatorname{fin}}
\renewcommand{\c}{{\mathrm c}}
\newcommand{\Cco}{C_\c}

\newcommand{\uk}{\underline{k}}

\newcommand{\lam}{\lambda}

\newcommand{\FF}{\mathcal{F}}

\newcommand{\hh}{\mathfrak{h}}

\DeclareFieldFormat{postnote}{#1}
\DeclareFieldFormat{multipostnote}{#1}

\title{Ground States for Infrared Renormalized   Translation-Invariant Non-Relativistic QED}

\author{David Hasler\footnote{Institut für Mathematik, Friedrich-Schiller-Universität Jena, Ernst-Abbe-Platz 2, 07743 Jena, Germany, \texttt{ \href{mailto:david.hasler@uni-jena.de}{david.hasler@uni-jena.de}}} \and Oliver Siebert\footnote{FB Mathematik, Universität Tübingen, Auf der Morgenstelle 10, 72076 Tübingen, Germany, \texttt{ \href{mailto:oliver.siebert@uni-tuebingen.de}{oliver.siebert@uni-tuebingen.de}}}}

\numberwithin{equation}{section}
\newtheorem{thm}{Theorem}[section]
\newtheorem{lemma}[thm]{Lemma}
\newtheorem{prop}[thm]{Proposition}
\newtheorem{proposition}[thm]{Proposition}

\theoremstyle{definition}

\theoremstyle{remark}
\newtheorem{remark}[thm]{Remark}
\newtheorem{rem}[thm]{Remark}

\crefname{thm}{Theorem}{Theorems}
\Crefname{thm}{Theorem}{Theorems}
\crefname{coro}{Corollary}{Corollaries}
\crefname{hyp}{Hypothesis}{Hypotheses}
\Crefname{hyp}{Hypothesis}{Hypotheses}
\crefname{lemma}{Lemma}{Lemmas}
\Crefname{lemma}{Lemma}{Lemmas}
\crefname{prop}{Proposition}{Propositions}
\crefname{proposition}{Proposition}{Propositions}
\crefname{enumi}{}{}
\Crefname{enumi}{}{}
\creflabelformat{enumi}{#2(#1)#3}
\crefname{equation}{}{}
\Crefname{equation}{}{}

\begin{document}

\maketitle

\begin{abstract}
We consider a  translation-invariant Pauli-Fierz model describing  a non-relativistic  charged  quantum mechanical   particle  interacting with the  quantized electromagnetic field.  The charged particle 
may be spinless or have spin one half.   We  decompose the Hamiltonian with
respect  to  the  total  momentum  into   a  direct  integral of so called fiber Hamiltonians. We perform an infrared renormalization, in the sense of 
norm  resolvent convergence,  for each fiber Hamiltonian,  which has the physical 
interpretation of removing  an infinite photon cloud. 
We show that the renormalized fiber  Hamiltonians  have  a ground state for  almost all values for the total momentum with modulus less than one. 
\end{abstract}

\section{Introduction}

The infrared problem  is a   technical   obstacle in describing scattering of charged particles with a quantized radiation field, which has been thoroughly  investigated  in  the literature over  the last decades. As already pointed out in the seminal paper by Bloch and Nordsieck \cite{blonor:37} it is an artifact of the emergence  of a cloud of \textit{infinitely} many photons with a \textit{finite} amount of total energy in the asymptotic scattering states. While there appeared several perturbative methods like the one by Fadeev and Kulish \cite{kulfad:70} for the computation of divergence-free scattering transition probabilities, a mathematical precise theory was not available for a long time.  A first rigorous description of  so called dressed one-electron states,  i.e., the `plane-wave' states of an electron traveling with a cloud of photons at some fixed total momentum, was given by Kibble  \cite{kibble} based on ideas of Chung \cite{chung}. The basic principle  behind it is, that due to the  infinite number of photons, the states cease to exist in the original Hilbert space of the Fock representation. Instead, one has to use non-Fock coherent representations of the canonical commutation relations (CCR). A significant problem is that different total momenta lead to different inequivalent representations, making it hard to construct scattering states, cf.   \cite{F73} and  for recent work  on this subject   \cite{cadamuro2019relative}.

As a first step towards a more realistic model Fröhlich discussed so called dressed one-electron  states in the Nelson model \cite{F73,F74}, where a charged particle like an electron interacts with a bosonic field via a coupling term which is linear in the creation and annihilation operators. He shows that it is impossible to implement   one-electron states for different total momenta in a joint Hilbert space, and he  gave a non-constructive proof for the existence of the corresponding coherent states as a weak-* limit where  the infrared regularization is removed. Later Pizzo accomplished an explicit perturbative construction of such states in a Hilbert space by means of a dressing transformation and an iterative algorithm  \cite{piz03}, which he then used for the construction of asymptotic  scattering states \cite{piz05}. Related results for    the weak coupling model of non-relativistic QED, which is technically more involved due to the  interaction being \textit{quadratic} in the creation and annihilation operators were obtained in  \cite{chenfroehlich,ChenFroehlichPizzo1,ChenFroehlichPizzo2,BCFS05,C01}.
Furthermore, in the last years there appeared a series of papers for Coloumb scattering in the Nelson model beetween two electrons starting with  \cite{coulomb_scattering_nelson1} and between an electron and atom \cite{coulomb_scattering_nelson4} in the presence of a quantized radiation field. Recently, there was also a novel approach developed for the construction of scattering states in the Nelson model without any infrared limiting procedure \cite{infraparticle_states_revisited}.

Due to its significance for scattering theory, a lot of progress  was made to investigate such one-electron states without infrared regularization in different models, e.g., apart from the Nelson model and non-relativistic QED also for the UV-renormalized Nelson model \cite{nelson_uv_renormalized}  and the semi-relativistic Pauli-Fierz model \cite{konenberg2014mass}. 
 Those states can be described mathematically as ground states of a (renormalized) fiber Hamiltonian corresponding to fixed total momentum. 

In this paper we show the existence of one-electron states in the non-relativistic Pauli-Fierz model for almost all momenta with modulus less than one and for all values of the coupling constant. The charged particle may be spinless or have spin one half, where in the latter case have to assume an  energy inequality, cf.  \eqref{eq:groundstateen}.  Moreover, we show that they are  ground states of fiber-wise renormalized Hamiltonians. A similar result has been obtained for  the  UV  renormalized translation invariant Nelson model in \cite{nelson_uv_renormalized} following the strategy developed in \cite{piz03}. Moreover, we want to mention a related result for the Nelson model in \cite{arai_gs}.  
We use a compactness argument, originally tracing back to \cite{GriLieLos:01},  in the same way as in \cite{hassie:20,spinboson}, where we showed the existence of ground states at zero momentum. The proof is  non-constructive but also non-perturbative.  The infrared  renormalization is done via a  Gross transformation 
similarly  as performed   in \cite{nelson1964,piz03,piz05}  for the Nelson model or for non-relativistic QED in \cite{chenfroehlich}. 

 The main results 
are presented  in  \Cref{sec:main results}. In  \Cref{sec:outline proofs} we outline  the strategy of the proof. The renormalization of the fiber Hamiltonians is performed  in \Cref{sec:transformation}. Key elements of the proof of  the main result, stated in  \cref{thm:main result},  are an argument based on second order  perturbation theory  in the total momentum and   a photon number bound   together with  a bound  on derivatives 
of the photon momenta, which can be found in \Cref{sec:pertheory} and \Cref{sec:infrabounds}, respectively. Finally, \Cref{sec:maincompproof} contains the proof of the main result, where the  
 compactness argument is presented  and    previous estimates  are used to establish indeed compactness.

\newcommand{\pmi}{\phi_{m,i}}
\newcommand{\pmj}{\phi_{m,j}}
\newcommand{\tot}{\operatorname{tot}}

\section{Model}
\label{sec:model}

Let   $\hh$ be a complex Hilbert space.  We define  the symmetric Fock space
\[
\FF(\hh)  :=  \bigoplus_{n=0}^\infty  \hh^{(n)} ,
\]
with $\hh^{(0)} := \C$ and $\hh^{(n)} :=  \mathcal{P}_n( \bigotimes_{k=1}^n \hh ) $,   $n \in \N$,
with   $\mathcal{P}_n$ denoting  the orthogonal projection onto the subspace of
totally symmetric tensors. Thus we can identify
$\psi  \in \FF(\hh) $ with the sequence $(\psi_{(n)})_{n \in \N_0}$ with $\psi_{(n)} \in \hh^{(n)}$.
The vacuum  is the vector $\Omega :=
(1,0,0, \ldots ) \in \FF(\hh)$. For a linear subspace $\mathfrak{v}$ of the Hilbert space $\hh$ we define $\FF_{\fin}(\mathfrak{v}) \subset \FF(\hh)$ as the vector space of  finite linear combinations of $\Omega$ and  
vectors of the the form $\mathcal{P}_n( v_1 \otimes \cdots \otimes v_n )$ with  $v_1,...,v_n \in \mathfrak{v}$ and $n \in \N$.

We define for $f \in \hh$ the creation operator $a^*(f)$ acting on vectors $\psi \in \FF(\hh) $ by
$$
(a^*(f) \psi)_{(n)} = \sqrt{n} \mathcal{P}_n ( f \otimes \psi_{(n-1)} )
$$
with domain $D(a^*(f)) := \{ \psi \in \FF(\hh)  : a^*(f) \psi \in \FF(\hh)  \}$. This  yields a densely defined
closed operator. For $f \in \hh$ we define the annihilation $a(f)$ as the adjoint of $a^*(f)$, i.e.,
$$
a(f)  = \left[a^*(f) \right]^* .
$$
 It follows from the definition  that
 $a(f)$ is anti-linear, and $a^*(f)$ is linear in $f$. Creation and annihilation operators  are well
 known to satisfy the so called canonical commutation relations
\begin{align}
\label{eq:ccr}
[ a^*(f) , a^*(g) ] = 0 \quad , \quad [a(f) , a(g) ] = 0 \quad ,
\quad [a(f) , a^*(g) ] =  \inn{ f , g }_\hh \; ,
\end{align}
where $f , g \in \hh$, $[\cdot, \cdot]$ stands for the commutator, and    $\inn{ f , g}_\hh$ denotes the
inner product of $\hh$. 
For $f \in \hh$ we introduce the following notation for the field operator and the conjugate 
field operator 
\begin{align*}
\phi(f) &=  \text{closure of }  \quad  \frac{1}{\sqrt{2}}  (a(f) + a^*(f) ), \\ 
\pi(f) & =    \text{closure of }  \quad    \frac{-\i}{\sqrt{2}}  (a(f)-a^*(f))  . 
\end{align*} 
We note that $\pi(f) = \phi(\i f) $.

For a self-adjoint operator  $A$  in $\hh$ we define  the operator $\dG(A)$ as
follows.   In $\hh^{(n)}$  we set
 $$
A^{(n)} := A \otimes \one \otimes \cdots \otimes  \one  + \one \otimes A \otimes \cdots \otimes \one + \cdots + \one \otimes \cdots \otimes \one \otimes A ,  \, n \in \N ,
$$
in the sense of \cite[VIII.10]{rs1} and  $A^{(0)} := 0$.
By definition   $\psi \in \FF(\hh) $ is  in the domain of
$\dG(A)$  if  $\psi_{(n)} \in D(A^{(n)})$ for all  $n \in \N_0$  and
\begin{align} \label{eq:defofgammaA}
(\dG(A) \psi )_{(n)} = A^{(n)} \psi_{(n)}   ,  \quad n \in \N_0 ,
\end{align}
is  a vector in $\FF(\hh)$,  in which  case $\dG(A) \psi $ is defined by \eqref{eq:defofgammaA}.
The operator $\dG(A)$ is self-adjoint, see for example \cite[VIII.10]{rs1}.

Henceforth, we shall consider specifically
\begin{equation} \label{eq:defofhilconc}
\hh := L^2(\Z_2  \times \R^3 ) \cong L^2(\R^3 ; \C^2 )
\end{equation}
and write $\FF$ for $\FF(\hh)$.
The Hilbert space $\hh$ describes  so called  transversally polarized photons. By physical interpretation
the variable $(\lambda,k) \in  \Z_2  \times \R^3$ consists of the wave
vector $k$ and the polarization label  $\lambda$.
Because of   \eqref{eq:defofhilconc}, the   elements  $\psi \in \FF_0$ can be identified
with sequences $(\psi_{(n)})_{n=0}^\infty$ of so called $n$-photon wave
functions, $\psi_{(n)}
\in L^2_{\mathrm{sym}}(( \Z_2 \times \R^3)^n)$, where the subscript ``${\mathrm{sym}}$'' stands for the subspace
of functions wich  are totally symmetric in their $n$ arguments.
Henceforth, we shall make  use of  this identification without mention.
The Fock space inherits a scalar
product from $\hh$, explicitly
\begin{align*} 
& \inn{  \psi , \varphi }  \\
& =  \overline{\psi}_{(0)} \varphi_{(0)} +
\sum_{n=1}^\infty \sum_{\substack{ \lam_1, \ldots ,\lam_n \\ \in \{ 1 , 2 \} } } \int
\overline{\psi_{(n)}(\lam_1, {k}_1, \ldots , \lam_n, {k}_n)}
\varphi_{(n)}(\lam_1, {k}_1, \ldots , \lam_n, {k}_n ) d {k}_1 \ldots d
{k}_n \; .
\end{align*}

For $m \geq 0$ the field energy operator denoted by $\Hfm$
is  defined by
$$
\Hfm = \dG(\omega_m)  ,
$$
where $\omega_m \: \Z_2 \times \R^3 \to \R$,  $\omega_m(\lambda,k)  :=   \omega_m(k) := \sqrt{ m^2 + k^2}$.
 The operator of momentum $\Pf$ is defined as a three dimensional vector of operators, where
 the $j$-th component is given by
 $$
 (\Pf)_j := \dG(K_j) ,
 $$
 with $K_j \: \Z_2 \times \R^3 \to \R$,  $K_j(\lambda,k) = k_j$.

 The Hilbert space describing the
system composed of a charged particle with  spin $s \in \{ 0 , \frac{1}{2} \}$  and the quantized field is
\begin{equation} \label{eq:defofham} 
\HH_{\mathrm{full}} := L^2( \R^3 \times \Z_{2 s + 1} )  \otimes \FF \; .
\end{equation} 
We consider the Hamiltonian
$$
H_m =   \frac{1}{2}\left(  p  + e A (x)  \right)^2  + e  S \cdot  B(x)  + \Hfm \; ,
$$
with
\begin{align} \label{eq:44}
A_j(x)  = \phi(f_{A,j} E_x)  , \quad  
B_j(x) &=\phi(f_{B,j} E_x) , \quad j =1,2,3,
\end{align}
where 
we defined the following functions $E_y \: \R^3 \to \R^3$ by $k \mapsto e^{ - \i k \cdot y}$ for $y \in \R^3$, 
\begin{align*} 
f_{A,j}(k,\lambda) &  := \frac{\rho(k)}{\sqrt{\abs k}} \epsilon_\lambda(k), \quad 
 f_{B,j}(k,\lambda)  := -\i  [ k \wedge f_A(k,\lambda)]_j   ,  \quad k \in \R^3 ,  %
\end{align*}  
where the $\varepsilon_{\lambda}(k) \in \R^3$ are  so called polarization  vectors, depending
measurably on $\widehat{k} = k /|k|$,  such that $(\widehat k,
\varepsilon_{1}(k), \varepsilon_{2}(k))$ forms an orthonormal basis.
 For  the proof
we shall  make  an explicit choice  of the polarization vectors in   \eqref{eq:cheps}, below. 
We note that the $x$  in \eqref{eq:44} denotes the operator  of multiplication with the position coordinates of the first component of the tensor product    \eqref{eq:defofham}. Mathematically this amounts to the following definition. 
By means  of the unitary isomorphism $\HH_{\mathrm{full}} \cong L^2(\R^3 \times \Z_{2s+1} ; \FF)$  
one has the identity 
\begin{align} 
[ (A_j(x)  \psi ](x,s) = \phi(f_{A,j} E_x) \psi(x,s)  , \quad (x,s) \in \R^3 \times \Z_{2s+1}  \label{defofA}, 
\end{align} 
for $\psi \in L^2(\R^3 \times \Z_{2s+1} ; \FF)$ such that $\psi(x,s) \in D( \phi(f_{A,j} E_x))$ for all $(x,s) \in \R^3 \times \Z_{2s+1}$ and  \eqref{defofA} defines  again an element in 
$L^2(\R^3 \times \Z_{2s+1} ; \FF)$. Likewise we define $B_j(x)$. In the physics literature the operators  \eqref{eq:44} are often  
  written in terms  of so called operator valued distributions. This is  outlined  in the following  remark. 

\begin{remark}  Introducing operator valued distributions $a_\lambda(k)$ and $a_\lambda^*(k)$ satisfying 
the so called canonical commutation relations,  \cite[X.7]{ressim:fou}, one can write 
\begin{align} \label{eq:44r}
A(x) & = \sum_{\lambda=1,2} \int \frac{ \varepsilon_{\lambda}(k) }{\sqrt{2|k|}}\left(
\overline{\rho(k)} a_{\lambda}(k) e^{\i k \cdot x}  +
\rho(k) a_{\lambda}^*(k) e^{-\i k \cdot x}  \right) dk
\; ,  \\
B(x) & = \sum_{\lambda=1,2} \int \frac{  \i k \wedge \varepsilon_{\lambda}(k)  }{\sqrt{2|k|}} \left(
\overline{\rho(k)} a_{\lambda}(k) e^{\i k \cdot x}  -
\rho(k) a_{\lambda}^*(k) e^{-\i k \cdot x} \right) dk
\; ,
\end{align}
where the integrals are understood as weak integrals on a suitable dense subspace. 
\end{remark} 

We shall adopt the standard convention that for  $v = (v_1,v_2,v_3)$  we write $v^2 :=   \sum_{j=1}^3 v_j v_j $.
By $x$ we denote the position of the electron and its canonically
conjugate momentum by $p = - \i \nabla_x$.
If $s=1/2$, let $S = (\sigma_1,\sigma_2,\sigma_3)$   denote the vector of  Pauli-matrices.
If $s=0$, let $S = 0$.
The number $e \in \R$ is called the coupling constant.
The so called form factor $\rho \: \R^3 \to \C$ is a measurable function for which we shall  assume  the following
 hypothesis for the main theorem. 
We assume that  for some  $\Lambda \in (0, \infty)$ we have
\begin{equation}  \label{e:formfac1}   %
\rho(k) = \frac{1}{(2 \pi)^{3/2}} {\mathbf 1}_{[0,\Lambda]}(|k|) \; ,   \quad k \in \R^3 .
\end{equation}
where ${\mathbf 1}_{[0,\Lambda]}$ denotes the characteristic function of the set $[0,\Lambda]$.
One can show  that $H_m$ is self-adjoint on $D(\Delta \otimes \Id) \cap D(\Id \otimes \Hfm)$, cf. \cite{HH08selfadjoint, H02}.

The
Hamiltonian is translation invariant and commutes with the generator
of translations, i.e., the operator of total momentum
$$
P_{\mathrm{tot}} = p + \Pf  \; .
$$
Let
$$
W = \exp(\i x \cdot \Pf) \; .
$$
Note $W P_{\mathrm{tot}} W^* = p $ so that in the new representation  $p$
is the total momentum. One easily  computes
$$
W H_m W^* =  \frac{1}{2}\left( p - \Pf  +  e A \right)^2   + e  S \cdot  B         +      \Hfm \; ,
$$
where we set $A := A(0)$ and $B := B(0)$. Let $F$ be the Fourier transform in the electron
variable $x$, i.e., on $L^2(\R^3)$,
\begin{eqnarray}
\label{eq:fourier} ( F \psi)(\xi ) = \frac{1}{(2\pi)^{3/2}}
\int_{\R^3} e^{-\i \xi \cdot x } \psi(x) dx \; .
\end{eqnarray}
Then the composition $U=FW$ is a  unitary operator 
\begin{align*}
 U \: \HH_{\mathrm{full}}  \to
L^2(\R^3  \times \Z_{2s+1}) \otimes \FF \cong L^2(\R^3 ; \C^{2s+1} \otimes \FF  ) = \int_{\R^3}^\oplus \C^{2s+1} \otimes   \FF  d\xi,
\end{align*} 
 yielding  the so called fiber decomposition of the
Hamiltonian,
\begin{align*}
U H_m U^* = \int_{\R^3}^{\oplus} H_m(\xi) d\xi , 
\end{align*} 
where 
\begin{equation} \label{def:ofhm}
H_m(\xi)  = \frac{1}{2}(\xi - \Pf + eA)^2 + e S \cdot B + \Hfm
\end{equation}
is an operator in  the so called  reduced  Hilbert space 
\begin{align*}
\HH :=  \C^{2s+1} \otimes \FF ,
\end{align*}
 cf. \cite{ressim:ana,SPOHN04}. %
 The following result shows that the operator \eqref{def:ofhm} is well defined, cf. \cite{H06},  \cite{HL08}, \cite{LMS06} and \cite{hassie:20}.

\begin{thm} \label{hyp:op0}  For all $m \geq 0$,     $\xi \in \R^3$, 
and    $e \in \R$  the operator   $H_m(\xi)$ is bounded from below and 
 self-adjoint on the natural domain of  $  \Pf^2 + \Hfm$.
\end{thm}

We call   
$$
E_m(\xi)  := \inf \sigma(H_m(\xi)) .
$$ the ground state energy, which in general does not need to be an eigenvalue.
Let us collect a few elementary properties of the ground state energy as a function of $\xi$.

\begin{lemma}  \label{eq:convex} The following holds. 
\begin{itemize}
\item[(i)]   The function $t_m \: \xi \mapsto  \frac{\xi^2}{2} -  E_m(\xi)$ is convex.  
\item[(ii)]   The function $\xi \mapsto E_m(\xi)$ is almost everywhere differentiable.  
\item[(iii)]   The function $\xi \mapsto E_m(\xi)$ is rotationally invariant. 
\item[(iv)]  If  $E_m$ is differentiable in $\xi$  we have $|\nabla E_{m}(\xi)| \leq |\xi|$  
\end{itemize} 
\end{lemma} 
\begin{proof} (i) This follows since the supremum of convex functions is convex. (ii)
Convex functions are locally Lipschitz, see for example Lemma  \ref{lem:eleconv} in the Appendix.
 Lipschitz functions  are almost everywhere differentiable by Rademacher's theorem,
see for example  Theorem  \ref{thm:rade} in the Appendix.  So (ii) follows from (i).
 (iii) Follows from well known transformation properties 
of the field operators and the rotation invariance of  $\rho$. (iv)  
By restricting the function $t_m$ onto  a straight line through the origin the claim 
  follows from  \cref{estonderconv}  in the appendix, where the symmetry assumption  
follows from (iii). 
\end{proof}

 Henceforth, we shall write $\Hf$, $H$, $H(\xi)$, and $E(\xi)$  for $H_{\f,0}$, $H_0$,  $H_0(\xi)$,  and $E_0(\xi)$, respectively.
For massless photons, wich corresponds to the case   $m=0$, the fiber Hamiltonian does not have a ground 
 state  for momenta $\xi$ for 
 which  $\nabla E(\xi) \neq 0$.  This is the content of the following theorem  shown in  \cite{haslerherbst1}.

\begin{thm} \label{thm:absgs}  Let $e \neq 0$. If $E(\cdot)$ is differentiable at $\xi$ and has a nonzero 
derivative, then $H(\xi)$ does not have a ground state. 
\end{thm} 

The physical interpretation of Theorem  \ref{thm:absgs}    is that charged 
 particles with nonzero velocity $\nabla E(\xi) \neq 0$, aquire  an infinite photon cloud which
 ceases to be square  integrable. This will be made mathematically more precise below and      can be 
 viewed as a   manifestation  of the so called  infrared catastrophe. 
On the other hand, if one considers the case with zero momentum $\xi = 0$
one can  show in fact  that the fiber Hamiltonian has a ground state.
This  has been established 
for small values of the coupling constant  in  \cite{C01}  and recently \cite{hassie:20}   for all values of the coupling constant
under an energy inequality assumption  \eqref{eq:groundstateen}, which is discussed below.  

To control the infrared divergence for $\xi \neq 0$ one typically  introduces an infrared
regularization. One  way to achieve this, is to introduce a  positive photon mass, i.e., one 
considers $H_m(\xi$) for $m > 0$.  One can  shown that 
  the operator   $H_m(\xi)$ has a ground state for all $m > 0$ and  $|\xi|$ sufficiently small.
Such a result has been obtained in various situations   \cite{F74,F73,SPOHN04,GriLieLos:01,LMS06}. 
In the following we will work with   the  result from  \cite{hassie:20}, since it suits our specific situation.
To formulate it,  we will use  the following  energy inequality, which  reads as follows. We say that the  energy inequality holds for  $e \in \R$ 
and $m \geq 0$ if 
 \begin{eqnarray} \label{eq:groundstateen} 
  E_m(\xi) \geq E_m(0), \quad \quad \forall \xi \in \R^3 .
\label{eq:eineq}
\end{eqnarray}

This inequality  has been intensively investigated  in the literature. 
In the spinless case $s = 0$  it has been shown to hold for all values 
 $m \geq 0$ and   $e \in \R$
using   functional integration,  \cite{G72,SPOHN04,H06, LMS06}. This is the content of the following theorem. 
\begin{thm} In the spin-less case, $s=0$,  \label{thm:nospingroundstateenineq}  
\eqref{eq:groundstateen}
holds for all $e \in \R$ and $m \geq 0$. 
\end{thm} 
In case of spin one half  it has been shown to hold, but only in a limited range of parameters. That is,  
for  $s=1/2$  Inequality   \eqref{eq:eineq}    has  to the best 
of our knowledge  not yet been  shown by means of  functional integration. 
 For $s=1/2$    Inequality \eqref{eq:eineq} follows for  small $|e|$  from the main theorem 
stated in \cite{C01}, which in turn  is based on  perturbative arguments.
Now let us  state the result about existence of ground states for positive photon mass  from   \cite{hassie:20}, which will be used in this paper, 
see also   \cite{F74,F73,SPOHN04,GriLieLos:01,LMS06}.

\begin{thm} 
\label{hyp:op1}
Let $e \in \R$ and  $m > 0$ and 
suppose  the energy inequality  \eqref{eq:eineq} holds.  
If   $|\xi| \leq 1$, then   $E_m(\xi)$ is an eigenvalue of   $H_m(\xi)$  isolated from the
essential  spectrum. 
\end{thm}

Henceforth,  we let    $\psi_m(\xi)$ denote a   normalized  eigenvector of $H_m(\xi)$ with eigenvalue $E_m(\xi)$, whose existence  is granted by \cref{hyp:op1}.  In  
 \cref{eq:propenegyconv},  in   \Cref{sec:outline proofs},  we will show that $E_m(\xi)$ is monotonically decreasing as $m \downarrow 0$ and 
\begin{equation} \label{eq:convofenergy} 
E_m(\xi) \to E(\xi)  .
\end{equation} 
As the positive mass is removed, the ground state $\psi_m(\xi)$  acquires an infinite photon 
cloud which  ceases to be square integrable  and hence  does not converge  in the Hilbert space. 
However, if one removes  this  diverging  photon 
cloud by a Bogoluibov  transformation,  one can show that the transformed ground 
states  converge to a nonzero vector in Hilbert space. Such a  picture has been established in various situations \cite{F73,F74,piz05,chenfroehlich,ChenFroehlichPizzo2}.
To describe this in more detail  we  define 
for  $m \geq  0$  the function 
\begin{equation}    \label{eq:defofh} 
h_{m,\xi}(\lambda,k) =   e \epsilon_\lambda(k) \cdot \nabla_\xi E_m(\xi) \frac{\rho(k)}{\sqrt{\abs{k}}} \frac{1}{\omega_m(k) - k \cdot \nabla_\xi E_m(\xi)} ,
\end{equation} 
where $(\lambda,k) \in  \{1 , 2\} \times  \R^3$. 
In Remark \ref{derivationcor}  we sketch a    heuristic argument   for  \eqref{eq:defofh}.

\begin{remark}\label{derivationcor} Following \cite[Section 2]{piz03} and references 
therein we give a heuristic and formal argument for the choice  \eqref{eq:defofh}.   Assuming that $\psi_m(\xi)$ should be  ``coherent in the
infrared region'' we make  the following   ``Ansatz''  as  $k \to 0$ 
\begin{align} \label{eq:ansatz87} 
a_\lambda(k) \psi_m(\xi) \approx - g_{m,\xi}(\lambda,k) \psi_m(\xi)   
\end{align} 
for some  function $g_{m,\xi}$ with values in the linear maps of $\C^{2s+1}$. By a formal application of the virial theorem we  have  
 \begin{align} \label{eq:comm87} \inn{ \psi_m(\xi) , [H_m(\xi) , a_\lambda(k) ] \psi_m(\xi) }  =   0  .
 \end{align}   
Calculating the formal commutator $[H_m(\xi),a_\lambda(k)]$ by means of the so called pull-through formula, c.f. 
 \cite[Lemma A.1]{BachFroehlichSigal.1998a},   inserting the ansatz   \eqref{eq:ansatz87} into 
 \eqref{eq:comm87}, using the Feynman-Hellmann formula, c.f.  Lemma \ref{lem:per},  and solving for $g_{m,\xi}$ we find   
\[
g_{m,\xi}(\lambda,k)  \approx   e  \left(  \nabla_\xi E_m(\xi) \cdot f_A(\lambda,k)  + S \cdot f_B(\lambda,k) \right) \frac{1}{\omega_m(k) - k \cdot \nabla_\xi E_m(\xi)  + \frac{1}{2} k^2}
\]
as  $k \to 0$.
Dropping the higher order terms involving   $k^2$  and $f_B$   we  arrive at  \eqref{eq:defofh}. 
\end{remark}

From \cref{eq:convex}  we see that the  derivative of $E_m$   exists  almost everywhere.
Now if the derivative exists, then   $h_{m,\xi} \in \hh$ whenever $m > 0$ and $|\nabla E_m(\xi)|<1$, where the latter inequality   always holds if $|\xi|<1$
 by 
 Lemma   \ref{eq:convex}. However in case  $m=0$ we have  
$h_{0,\xi}  \notin \hh$. 
Whenever  $h_{m,\xi} \in \hh$,  we can  define  
\begin{align}
\label{eq:Um definition}
U_m(\xi) = e^{ \i \pi(h_{m,\xi})}.
\end{align}
Working  with  $U_m(\xi)$ will  require  
 control of the derivatives of the ground state energies. In Proposition \ref{prop:nabla_E_converging} at  the end of \Cref{sec:outline proofs}
we show that for  any seqeunce 
$(m_j)_{j \in \N}$  of nonnegative   numbers which tend  to zero, we have
 \begin{align} \label{eq:convenery} 
\nabla E_{m_j}(\xi) \overset{j \to \infty}\longrightarrow \nabla E(\xi) 
\end{align} 
for almost all $\xi \in \R^3$. This will allow us to  study strong  limits of  $U_{m}(\xi)$ as $m \downarrow 0$.

\section{Statement of Main Results} 
\label{sec:main results}

As a first result we will show in Section \ref{sec:maincompproof}  the following theorem, which describes explicitly   the removal of the divergent photon cloud of the ground 
state as the positive photon mass tends to zero.

\begin{thm} \label{thm:main222}     Let $e \in \R$ and 
suppose there exists an $m_0 > 0$ such that the   energy inequality  \eqref{eq:eineq} holds for 
all $m \in (0,m_0)$.   Then for almost  all  $\xi \in \R^3$ with $|\xi| < 1$,
 there exists a sequence $(m_j)_{j \in \N}$ of positive  numbers  converging to zero
such that  $U_{m_j}(\xi) \psi_{m_j}(\xi)$ converges to a  nonzero vector, $\widehat{\psi}_0(\xi)$,  in the Hilbert space $\HH$. 
\end{thm}

\begin{remark}  
We note that  Theorem   \ref{thm:main222}  is similar to a  result obtained in  \cite{ChenFroehlichPizzo2},
where  convergence  is shown for a  spinless charged particle.  The  proof  given in that paper   is constructive,
and    holds   for all $\xi$ with  $|\xi| <  1/3$  and   small values of $|e|$. 
 Theorem \ref{thm:main222} extends that  result   for  almost all $\xi$ with  $|\xi| < 1$ 
 to  all values of the coupling constant $e \in \R$, see also Remark  \ref{rem:setofxi}.
\end{remark}

Next we address the question whether the nonzero vector  $\widehat{\psi}_0(\xi)$ given in Theorem \ref{thm:main222}   can be expressed  as the ground state of an infrared renormalized  Hamiltonian.
The main result of  this paper is,  that this is indeed the case. For this  we will show in  Theorem   \ref{thm:existoftrafoham}, below,   
 that  as  $m_j$ tends to zero  the operators 
\begin{align} \label{conjHlim} 
U_{m_j}(\xi) H(\xi) U_{m_j}(\xi)^*   
\end{align} 
coverge in   resolvent sense for almost all $\xi$ with $|\xi| < 1$, to a self-adjoint operator, which will
be given explicitly in  \eqref{defofhhat}   below,  and which we will  refer to as the infrared renormalized fiber Hamiltonian.  
To capture  this limit in  explicit terms  we introduce first for measurable functions   $f,g \: \Z_2 \times \R^3 \to  \C$ with   $ \overline{f} g \in L^1(\Z_2 \times \R^3)$
 the sesquilinear form 
$$
\mathfrak{s}(f,g) = \frac{1}{2} \sum_{\lambda=1}^2  \int \overline{f}(\lambda, k) g(\lambda,k) dk   . 
$$
Second  for  a measurable function  $f \: \Z_2 \times \R^3 \to \R$  we define the 
following formal operator 
\begin{align}
T_m(f;\xi)  &:= \frac{1}{2}\sum_{j=1}^3 \left( \xi_j - P_{f,j} + eA_j - \phi({K_j} f) -  2e \Re  \mathfrak{s}(f ,f_{A,j}) - \mathfrak{s}(f, {K_j}  f ) \right)^2 \nonumber \\  &\quad + e\sum_{j=1}^3 S_j \left(B_j - 2 \Re \mathfrak{s}(f , f_{B,j}  )  \right)  +  \Hfm + \phi(\omega_m f) + \mathfrak{s}( f , \omega_m  f) .    \label{eq:defoftm} 
\end{align} 

We   will  show in  in Section \ref{sec:trafham},   Lemma  \ref{lem:asslem},  that this operator is  $\Pf^2 + \Hfm$ bounded, provided
the function $f$ satisfies certain  properties. Specifically, 
using  the next  lemma,  following from  
\cref{convham}   in \Cref{sec:trafham}, 
 we can  give an explicit definition of  the   infrared renormalized fiber Hamiltonian. 
\begin{lemma} \label{th:renormalized self-adjoint} 
	 Let $E$ be differentiable in $\xi$ and suppose $|\nabla_\xi E(\xi)|<1$. 
Then  
\begin{align} \label{defofhhat}  \widehat{H}(\xi) := T_0(h_{0,\xi};\xi) \end{align} 
can be realized  as a    selfadjoint operator  with domain $D(\Pf^2 + \Hf)$.
\end{lemma}

We now state the first main result of this paper, which follows from \cref{convham}   in \Cref{sec:trafham} as well.
\begin{thm} \label{thm:existoftrafoham} Let $E$ be differentiable in $\xi$ and suppose $|\nabla_\xi E(\xi)|<1$. Let  $(m_j)_{j \in \N}$ be  a sequence of positive  numbers converging to zero  and 
 $\nabla E_{m_j}(\xi) \to \nabla E(\xi)$.  
Then the 
operators  
$$U_{m_j}(\xi) H(\xi) U_{m_j}(\xi)^* , $$   converge
 as $j \to \infty $  in 
norm resolvent sense to  $\widehat{H}(\xi)$. %
\end{thm}

We note that in view of Lemma  \ref{eq:convex} and  \eqref{eq:convenery} the assumptions of Theorem   \ref{thm:existoftrafoham}  are 
valid for almost all $\xi \in \R^3$ with $| \xi | < 1$. 
We now state   the second  main result of this paper, which will be shown in  Section \ref{sec:maincompproof}, and  which covers   charged particles without  spin   and 
 charged particles with spin one half.

\begin{thm}\label{thm:main result}
 Let  $e \in \R$.    Suppose  that one of the following two assumptions is satisfied:
 \begin{itemize}
 \item[(i)] $s=0$,
 \item[(ii)] $ s=1/2$ and  
 there exists an $m_0 > 0$ such that the energy inequality 
  \eqref{eq:eineq} holds for all $m \in (0,m_0)$. 
  \end{itemize} 
Then for  almost all $\xi \in \IR^3$  with $\abs {\xi} <  1$ the following holds.
\begin{itemize}
\item[(a)] 
 The function  $E$ is differentiable at $\xi$,
  $|\nabla E(\xi)|<1$, and 
 the operator 
$\widehat{H}(\xi)$ has a ground state, i.e., $E(\xi)$ is an eigenvalue of $\widehat{H}(\xi)$.
\item[(b)]  There exists a sequence $(m_j)_{j \in \N}$ of positive  numbers  converging to zero
such that  $U_{m_j}(\xi) \psi_{m_j}(\xi)$ converges to the ground state.
\end{itemize} 
\end{thm}

The  fact that   the assertion of  Theorem \ref{thm:main result} only holds for almost all $\xi$ with length less than one, does not 
affect  the fiber direct integral, %
since   sets of measure zero 
do not play a role in integration.

\begin{rem}\label{rem:setofxi} We  will determine  in \cref{thm:mainmain}, below,   the  set  of  $\xi$'s,  for which the assertions of  \cref{thm:main222} as well as \cref{thm:main result} can be shown, more precisely.  That  set only  depends on regularity properties of  the energies $E_m$. 
We note that  regularity properties for such 
models have been studied   using various methods. In this regard we want to mention well established methods from renormalization,    iterated perturbation, or  statistical mechanics, cf.    \cite{C01,BCFS05,piz03,AbdessalamHasler.2011}. 
\end{rem}

\section{Outline of the Proofs} 
\label{sec:outline proofs} 

Theorem \ref{thm:existoftrafoham} will be shown in  Section \ref{sec:trafham} using 
well established  estimates involving creation and annihilation operators. 
To prove the existence of a ground state for the renormalized Hamiltonian
 $\widehat{H}(\xi)$, i.e., \cref{thm:main result}, we  follow  closely ideas  given  in \cite{GriLieLos:01} and \cite{LMS06} combined with a regularization 
procedure used in  \cite{F73,F74,piz05,ChenFroehlichPizzo2}.

The basic idea of the proof is to regularize the Hamiltonian by adding  a positive mass to the photons
 and to study  ground state properties 
when this regularization is removed. That is, 
we  consider  the Hamiltonian $H_m(\xi)$  for  $m > 0$.  Recall that 
  the operator   $H_m(\xi)$ has a normalized  ground state for all $m > 0$ and  $|\xi|$ sufficiently small,
see  Theorem \ref{hyp:op1}, which we  denote by  $\psi_m(\xi)$.  %
 In \cref{eq:propenegyconv} we will show that $E_m(\xi)$ is monotonically 
decreasing as $m \downarrow 0$ and 
\[
E_m(\xi) \to E(\xi)  .
\]

Working  with  $U_m(\xi)$  requires  
 control of the derivatives of the ground state energies. In Proposition \ref{prop:nabla_E_converging} at the end of  this section 
we show that for  any seqeunce 
$(m_j)_{j \in \N}$  of nonnegative numbers which converges to zero, 
   there exists a set $D$ such that $\R^3 \setminus D$ has Lebesgue measure zero,
such that   $E_{m_j}$, $j \in \N$,  and $E$ are differentiable at all points in $D$,  and 
 for all $\xi \in D$ we have 
 \begin{align} \label{convofder2} 
\nabla E_{m_j}(\xi) \overset{j \to \infty}\longrightarrow \nabla E(\xi) .
\end{align} 
Using  \eqref{convofder2}  we will show in  Proposition \ref{convham}  in Section \ref{sec:trafham}  that 
$$
U_{m_j}(\xi) H(\xi) U_{m_j}(\xi)^* \to \widehat{H}(\xi)  
$$
in norm resolvent sense for almost all $\xi$ with $|\xi| < 1$.  We then show in  Proposition    \ref{convofenergyexp} using 
\eqref{eq:convofenergy}  and standard estimates involving creation and annihilation operators 
  that  $U_{m_j}(\xi) \psi_{m_j}(\xi)$ is a minimizing sequence for $\widehat{H}(\xi)$, i.e.,  
\begin{align}  \label{convofenergyexp0} 
0 \leq \inn{ U_{m_j}(\xi) \psi_{m_j}(\xi) , 
( \widehat{H}(\xi) - E(\xi)  ) U_{m_j}(\xi) \psi_{m_j}(\xi) } \to 0  . 
\end{align} 
The key idea for proving the existence of a ground state is  to show  that  $U_{m_j}(\xi) \psi_{m_j}(\xi)$ has a convergent subsequence.  This 
is achieved  by showing that the sequence of these vectors lies in a  compact subset of
the Hilbert space. To this end,  we make use of the so called pull-through formula, 
which is shown   in  Lemma \ref{lem:pullthrough},  
 \begin{equation} \label{eq:pullthroughrel0} 
a_\lambda(k) \psi_m(\xi) =  \frac{e \rho(k)}{\sqrt{2 \abs k }}  R_{m,\xi}(k)\left(- \epsilon_\lambda(k)  \cdot v(\xi) +  S \cdot (\i k \wedge \epsilon_\lambda(k)) \right) \psi_m(\xi) , 
\end{equation} 
where  we defined 
\begin{align}
R_{m,\xi}(k) & := (
  H_m(\xi -k) + \omega_m(k) - E_m(\xi)   )^{-1}   ,  \label{defofR} \\
v(\xi)   &  := \xi - \Pf + e A . \label{defofv} 
\end{align} 
To estimate the resolvent occuring on the right hand side of \eqref{eq:pullthroughrel0}, we will relate 
it to the second order derivative of the  ground state energy as a function of $\xi$. 
 In Proposition \ref{prop:nabla_E_converging}  estimates 
on the second order derivatives of the ground state energy are established using convexity properties,
 where it is shown that for almost all $\xi$ the second order derivatives  $\partial_l^2 E_{m_j}(\xi)$  exist and 
$$
 \liminf_j (-  \partial_l^2 E_{m_j}(\xi)) < \infty 
$$
  for every $l=1,2,3$.  In Section \ref{sec:pertheory}
it  will be   shown  that these bounds on the second order derivative of 
the energy yield  bounds on the resolvent. Using these resolvent bounds 
in the pull-through formula \eqref{eq:pullthroughrel0} 
we  will derive  in Propositions 
\ref{prop:a_est_2}   and  \ref{lem:second2}
of 
 Section \ref{sec:infrabounds} estimates  on 
$$ a_\lambda(k) U_{m_j}(\xi) \psi_{m_j}(\xi) \text{  and  } \nabla_k a_\lambda(k) U_{m_j}(\xi)\psi_{m_j}(\xi) . $$  With the help of   these estimates we will   prove a fractional derivative bound in 
Lemma  \ref{thm:compact}. 
In Section   \ref{sec:maincompproof}  we use this fractional derivative bound to show that the vectors  $U_{m_j}(\xi)\psi_{m_j}(\xi)$, $j \in \N$,  lie 
indeed in  a compact set.  Thus by compactness  there exists  a  strongly convergent subsequence which 
converges to a nonzero vector, say  $\widehat{\psi}_0(\xi)$.  
We recall that this  is the content  of  Theorem  \ref{thm:main222} , which will be proven   in Section   \ref{sec:maincompproof}.

Using lower semicontinuity of nonnegative quadratic forms \cite{Sim:77}  (or alternatively the spectral theorem and  Fatou's Lemma),
we will conclude  from \eqref{convofenergyexp0} and Theorem \ref{thm:main222}  
that for almost all $\xi$ with $|\xi|< 1$ 
\begin{align*}
 0  & \leq \inn{  \widehat{\psi}_0(\xi) , ( \widehat{H}(\xi)  - E(\xi) ) \widehat{\psi}_0(\xi) }  \\
& \qquad \leq \liminf_{j \to \infty} \inn{ U_{m_j}(\xi)  \psi_{m_j}(\xi) , ( \widehat{H}(\xi) - E(\xi) )  U_{m_j}(\xi) \psi_{m_j} } = 0 , 
\end{align*}
i.e., that $\widehat{\psi}_0(\xi)$ is a ground state of $\widehat{H}(\xi)$. 
As  outlined in Section \ref{sec:maincompproof}
 this will then establish \cref{thm:main result}.

In the remaining part of this section we derive further properties of the ground state energies $E_m(\xi)$.

\begin{prop} \label{eq:propenegyconv} Let  $\xi \in \R^3$. Whenever $m_1 \geq m_2 \geq 0$ 
we have  $E_{m_1}(\xi)   \geq E_{m_2}(\xi) \geq E(\xi)$, and
\begin{eqnarray} \label{eq:enegineq1}  E(\xi) = \lim_{m \downarrow  0} E_m(\xi) \; . \end{eqnarray}
\end{prop}
\begin{proof}
For $0 \leq m_2 \leq m_1$ we have  $\omega \leq \omega_{m_2} \leq \omega_{m_1}$
and hence   $H(\xi) \leq H_{m_2}(\xi) \leq H_{m_1}(\xi)$.This implies the first statement. Moreover, 
it  implies the existence of the limit
\begin{equation} \label{eq:emconv}
E(\xi) \leq \lim_{m \downarrow 0} E_m(\xi)   .
\end{equation}
To show the opposite inequality we argue as follows. From Theorem  \ref{hyp:op0}
it follows that any core for $\Pf^2 + \Hf$ is a core for $H(0)$. 
Thus for any $\epsilon
>0$, there exists a vector $\phi \in D(N) \cap D(\Pf^2 + \Hf)$ such that
$$
\inn {\phi, H(\xi) \phi } \leq E(\xi) + \epsilon \; .
$$
On the other hand, since  $H_m(\xi) \leq H(\xi) + m \dG(\Id)$, it follows that for any
$m$,
$$
E_m(\xi) \leq \inn{\phi, H_m(\xi) \phi} \leq \inn{\phi, H(\xi) \phi } + m \inn{\phi , \dG(\Id)
\phi} \leq E(\xi) + \epsilon + m \inn{\phi, \dG(\Id)  \phi} \; .
$$
Hence
\begin{equation} \label{eq:emconv2}
 \lim_{m \downarrow  0} E_m(\xi) \leq E(\xi) + \epsilon .
\end{equation}
 Since
$\epsilon > 0$ is arbitrary,  \eqref{eq:enegineq1}  follows from \eqref{eq:emconv} and  \eqref{eq:emconv2}. 
\end{proof}

\begin{prop}
\label{prop:nabla_E_converging}
Let  $(m_j)_{j \in \N}$  be a sequence of nonnegative numbers which converges to zero. 
 Then  there exists a set $D$ such that $\R^3 \setminus D$ has Lebesgue measure zero
and the following holds.   
The functions  $E_{m_j}$, $j \in \N$,  and $E$ are differentiable at all points in $D$,  and 
 for all $\xi \in D$ we have 
\begin{itemize}
\item[(a)]  $\nabla E_{m_j}(\xi) \overset{j \to \infty}\longrightarrow \nabla E(\xi)$,
\item[(b)]   the second partial derivatives  $\partial_l^2 E_{m_j}(\xi)$ exist and  satisfy for every $l=1,2,3$ that $\liminf_j (-  \partial_l^2 E_{m_j}(\xi)) < \infty$ . 
\end{itemize} 
\end{prop}
\begin{proof} We note that $t_m \: \xi \mapsto \frac{1}{2}\xi^2 - E_m(\xi)$ 
are convex functions by  Lemma     \ref{eq:convex}.
Together with Proposition  \ref{eq:propenegyconv}, we see that 
(a)  and (b) follow  from \cref{lemma:derivative_converging,lemma:onvex2david} 
in the appendix, respectively.  
\end{proof}

\newcommand{\T}{\overline T}
\newcommand{\hpsi}{\hat \psi}
\newcommand{\hpsio}{\hat \psi_0(\xi)}

\section{Transformation of the Hamiltonian}
\label{sec:transformation}

\label{sec:trafham}

 We will now look at  unitary transformations of the form 
\[
e^{\i\pi(f)} H_m(\xi) e^{-\i\pi(f)}  ,   \quad f \in \hh ,
\]
and  limits of such objects. 
Some background is collected in Appendix \ref{secexprea}.  
 
Before we state the lemma we define 
\[
\| f \|_{(m)} = \| f \| +  \| \omega_m^{-1/2}  f \| 
\]
for all measurable functions $f \: \R^3 \times \Z_2 \to \C$, where $\| \cdot \|$ denotes the $L^2$-norm in $L^2(\Z_2 \times \R^3)$.  We  write $f \in L^2_{(m)}(\Z_2 \times \R^3)$ if $\| f \|_{(m)} < \infty$.

\begin{lemma} \label{lem:asslem}    Let $m \geq 0$ and  $f \: \Z_2 \times \R^3 \to \R$  be a measurable 
function,  with 
\begin{align} 
&K_j f  , K_j^2 f, \omega_m f   \in L^2_{(m)}(\Z_2 \times \R^3)  ,  \label{eq:ass1}  \\
& \omega_m |f|^2,  K_j |f|^2,  e  f \overline{ f}_{A,j}, e f \overline{f}_{B,j}  \in L^1(\Z^2 \times \R^3) . \label{eq:ass2} 
\end{align}  
 Then $T_m(f;\xi)$,  given in   \eqref{eq:defoftm}, defines a  $ \Pf^2 +\Hfm$ bounded  operator on $D( \Pf^2 + \Hfm)$.
In particular, there exits a constant $C$ such that for all $\psi \in D(\Pf^2 +\Hfm)$ we have 
\begin{align} \label{eq:estontrafham}  
 \| T_m(f, \xi) \psi \| \leq &  C \sum_j \big\{  \|e  f_{B,j} \|_{(m)} + \|  e \overline{f} f_{B,j} \|_1 + \| \omega_m f \|_{(m)}  + \||f|^2 \omega_m \|_1   \\
 & + (  \| e f_{A,j} \|_{(m)} + \| K_j f \|_{(m)}  + |c_j(f)|+1)^2 \nonumber  \\
& +  \| K_j^2 f \|_{(m)} + \| e K_j f_{A,j} \|_{(m)} ) \big\}  \| (\Pf^2 + \Hfm + 1) \psi \|  ,   \nonumber 
\end{align}
where 
\begin{align} \label{defofcj} 
c_{j}(h)  &:= \xi_j  - 2e {\Re} \mathfrak{s}(h  ,f_{A,j}) - \mathfrak{s}( h , K_j h) .
\end{align} 
\end{lemma} 
\begin{proof}  From the  assumptions \eqref{eq:ass2} it follows that all expressions involving 
the sesquilinear form $\mathfrak{s}$ are finite. Multiplying out the square in  \eqref{eq:defoftm} we will show that each term 
is $\Pf^2 + \Hfm$ bounded. For $\psi \in \FF_{\fin}(\Cco(\Z_2 \times \R^3))$ we find 
that 
\begin{align}
& T_m(f;\xi)  \psi   \label{eq:pfsqhfm00}  \\
 &  =  \left(  \frac{1}{2} \Pf^2 +  \Hfm \right) \psi \nonumber \\
& - \frac{1}{2} \sum_j  \Pfj{j}  \left(   \phi( e f_{A,j} -   {K_j} f   ) + c_j(f)  ) \right) \psi  \label{eq:pfsqhfm00-1} \\
& -   \frac{1}{2} \sum_j  \left( \phi( e f_{A,j} -   {K_j} f   ) + c_j(f)   \right)  \Pfj{j} \psi 
 \label{eq:pfsqhfm00-2}  \\
&+ \frac{1}{2}\sum_{j} \left(  \phi(  e f_{A,j} -   {K_j} f   ) + c_j(f)  \right)^2 \psi 
\label{eq:pfsqhfm00-3}  \\  & + e\sum_{j} S_j \left( \phi(f_{B,j})  - 2 {\Re} \mathfrak{s}(f , f_{B,j}  )  \right)  \psi  + ( \phi(\omega_m f) + \mathfrak{s}( f , \omega_m  f)) \psi  .    \label{eq:pfsqhfm00-4} 
\end{align} 
From the first identity in   Lemma  \ref{lem:basichfest}, see that the terms in    \eqref{eq:pfsqhfm00-4}   
 are  $\Hfm$ bounded, since  $ \omega_m f, f_{B,j} \in L^2_{(m)}$, i.e., 
\begin{align*} 
  \|  \eqref{eq:pfsqhfm00-4}\|  \leq C \big( \sum_j ( \|e  f_{B,j} \|_{(m)} + \| e  \overline{f} f_{B,j} \|_1 + \| \omega_m f \|_{(m)}  + \||f|^2 \omega_m \|_1
\big) \| (\Hfm + 1)^{1/2} \psi \|. 
\end{align*} 
Now let us estimate the terms in  \eqref{eq:pfsqhfm00-3}.
Clearly $(\phi(K_j f) + e A_j)^2$ is $\Hfm$ bounded by Lemma  \ref{lem:basichfest},
since $K_j f , f_{A,j} \in L^2_{(m)}(\Z_2 \times \R^3)$ and so 
\begin{align*} 
  \| \eqref{eq:pfsqhfm00-3} \| \leq C  \sum_{j=1}^3 ( \| e f_{A,j} \|_{(m)} + \| K_j f \|_{(m)}  + |c_j(f)|+1)^2 \| (\Hfm + 1 ) \psi \| .
\end{align*} 
Similarly,  the  term in  \eqref{eq:pfsqhfm00-2} can be bounded   by Lemma  \ref{lem:basichfest} as 
\begin{align*} 
  \| \eqref{eq:pfsqhfm00-2} \| \leq C \sum_j \big(  \| e f_{A,j} \|_{(m)} + \| K_j f \|_{(m)}  + |c_j(f)| \big) \| (\Hfm + 1 )^{1/2} \Pfj{j} \psi \| .
\end{align*} 
For 
\eqref{eq:pfsqhfm00-1} observe  that  \cref{lem:comdgamma} yields on $\FF_{\fin}(\Cco(\Z_2 \times \R^3))$
\begin{align}  \label{eq:commest} 
 P_{f,j}  \phi(K_j f + e f_{A,j} )   = \phi(K_j f + e f_{A,j} )   P_{f,j}  - \i  \phi( \i K_j^2 f + \i  e K_j f_{A,j} )  .
\end{align} 
Now each term in \eqref{eq:commest} is estimated using Lemma  \ref{lem:basichfest} and we find for   $\psi \in \FF_{\fin}(\Cco(\Z_2 \times \R^3))$ 
\begin{align} 
  \| \eqref{eq:pfsqhfm00-1} \| & \leq C \sum_j \big(  \| e f_{A,j} \|_{(m)} + \| K_j f \|_{(m)}  + |c_j(f)| \big) \| (\Hfm + 1 )^{1/2} \Pfj{j} \psi \| \nonumber  \\
& \quad + C \sum_j ( \| K_j^2 f \|_{(m)} + \| e K_j f_{A,j} \|_{(m)} )   \| (\Hfm + 1) \psi \|  .\label{eq:pfsqhfm00-11}
\end{align} 
Then  \eqref{eq:pfsqhfm00-11} extends to all $\psi \in D(\Pf^2 + \Hfm)$, since  $\FF_{\fin}(\Cco(\Z_2 \times \R^3))$ is a core for $\Pf^2 + \Hfm$.  
Finally, collecting above estimates shows \eqref{eq:estontrafham}  and hence the claim follows.  
 
\end{proof} 

We note that if the assumptions of Lemma \ref{lem:asslem} hold, we will define $T_m(f;\xi)$ 
according to that lemma as the operator with domain $D(\Pf^2 +\Hfm)$. 
 The following lemma will be  used for the proof of  Proposition \ref{prop:gross_transformed_hamiltonian}, below,  to establish 
self-adjointness. We note that the self-adjointness could alternatively be shown    directly using alternative methods, cf.  \cite{hassie:20,LMS06,H02} and references therein.

\begin{lemma} \label{eq:lemdomain} 
There exists a constant $c_0$ such that for 
 $m \geq 0$ and   $f \in \hh$ with   $ \omega^2 f  \in  L^2_{(m)}(\Z_2 \times \R^3 )$ 
we have   
\begin{align} \label{eq:trafoofpfhfm}  
& \left\| \left( \frac{1}{2} \Pf^2 + \Hfm \right)  e^{\i\pi(f)} \psi \right\|  \leq  C (f)  \|  (  \Pf^2 + \Hfm + 1  )\psi \|  ,
\end{align} 
where 
\begin{align*}
C(f)  & =  c_0  \big(  
\big[  (1 +  \| \omega_m^{1/2}  f \| + \| \omega_m f \|_{(m)} 
    \big]^2  + \| \omega_m^2 f \|_{(m)}   \big).  
\end{align*} 
 In particular, $ e^{\i\pi(f)} D( \Pf^2 + \Hfm) = D( \Pf^2 + \Hfm )$.
\end{lemma} 
\begin{proof}  We apply  
 \cref{lem:dgammaw} to the  %
operators
 $\Pfj{j}$ and $\Hfm$,
 and find  on $ \FF_{\fin}(  C_\c(\Z_2 \times \R^3)  ) \subset D( \Pf^2 +\Hfm)$
\begin{align}
&  e^{-\i\pi(f)}    (  \Pf^2 + \Hfm)  e^{\i\pi(f)}   \label{eq:pfsqhfm}  \\
& \quad =  \frac{1}{2} e^{-\i\pi(f)}  \Pf e^{\i\pi(f)}   \cdot e^{-\i\pi(f)}  \Pf  e^{\i\pi(f)}  +  e^{-\i\pi(f)}   \Hfm e^{\i\pi(f)}   \nonumber \\
 & \quad = \frac{1}{2} \sum_{j=1}^3 ( P_{f,j} + \phi(  K_j   f)  +  \frac{1}{2} \mathfrak{s}( f , K_j  f ))^2   + \Hfm + \phi( \omega_m f ) + \frac{1}{2} \mathfrak{s}(  f  , \omega_m f  )  .  \nonumber 
\end{align} 
Now the  bound   \eqref{eq:trafoofpfhfm}    follows 
from Lemma \ref{lem:asslem} with $e=0$, $\xi=0$, and $s=0$
using  the fact that  $\FF_{\fin}(C_\c(\Z_2 \times \R^3))$ is a core for $\Pf^2 + \Hfm$.  
 This implies that  $ e^{\i\pi(f)} D( \Pf^2 + \Hfm) \subset  D( \Pf^2 + \Hfm )$. 
Finally, the unitarity of  $e^{\i\pi(f)}$ implies that we have in fact equality.
\end{proof}

\begin{prop}
\label{prop:gross_transformed_hamiltonian}
Let $m \geq 0$, and let  $f \in \hs$  satisfy 
$K_j^2 f  \in L^2_{(m)}(\Z_2 \times \R^3) $, $j=1,2,3$. %
Then we have on $D(\Pf^2 +\Hfm)$
\muun
{
\label{eq:gross_transformed_hamiltonian}
e^{-\i\pi(f)} &H_m(\xi) e^{\i\pi(f)}  =  T_m(f;\xi) 
}
and the operator is self-adjoint on $D( \Pf^2 + \Hfm )$. 
\end{prop}
\begin{proof} 
We apply  \cref{lem:bratt} to the field operators $A$ and $B$ 
and 
 \cref{lem:dgammaw} to the  %
operators
 $\Pfj{j}$ 
 and  $\Hfm$, and find on $ \FF_{\fin}(  \Cco(\Z_2 \times \R^3)  ) \subset D( \Pf^2 +\Hfm)$
\begin{align}
&  e^{-\i\pi(f)}  H_m(\xi)  e^{\i\pi(f)}   \label{eq:pfsqhfm0}  \\
& \quad =  \frac{1}{2} e^{-\i\pi(f)}  (\xi - \Pf + e A) e^{\i\pi(f)}   \cdot e^{-\i\pi(f)}   (\xi - \Pf + e A) e^{\i\pi(f)}  + e S \cdot  e^{\i\pi(f)}     B  e^{\i\pi(f)}   \nonumber   \\
& \qquad \qquad +   e^{-\i\pi(f)}   \Hfm e^{\i\pi(f)}   \nonumber \\
 & \quad =  \frac{1}{2}\sum_{j=1}^3 \left( \xi_j - P_{f,j} + eA_j - \phi({K_j} f) -  2e {\Re}  \mathfrak{s}(f ,f_{A,j}) - \mathfrak{s}(f, {K_j}  f ) \right)^2 \nonumber \\ 
  &\quad \qquad + e\sum_{j=1}^3 S_j \left(B_j - 2 {\Re} \mathfrak{s}(f , f_{B,j}  )  \right)  +  \Hfm + \phi(\omega_m f) + \mathfrak{s}( f , \omega_m  f)  \nonumber  \\
& =
T_m(f;\xi)  \nonumber  ,
\end{align} 
where the  last identity is simply the definition  \eqref{eq:defoftm}  of $T_m(f;\xi)$.
By  Theorem    \ref{hyp:op0}    we know that $H_m(\xi)$ is self-adjoint on  $D(\Pf^2 + \Hf)$. 
By unitarity of  $e^{- \i \pi(f)}$ if follows that $e^{-\i\pi(f)} H_m(\xi) e^{\i\pi(f)}$ is self-adjoint 
on the domain $ e^{-\i\pi(f)}D(\Pf^2 + \Hf)$. But this domain equals $D(\Pf^2 + \Hf)$ by
Lemma  \ref{eq:lemdomain}. 
\end{proof}

Next we want to consider  limits  of expressions of the form  \eqref{conjHlim}.     For this we  prove the following lemma. 

\begin{lemma} \label{eq:thmconren2} Suppose $f$ satisfies the assumptions of 
Lemma \ref{lem:asslem} and let $g \in \hh$. Then   for all $\psi \in D(\Pf^2 + \Hf)$ we have 
\begin{equation} 
\label{eq:HHm_est0}
\nn{(T_0(f;\xi) -e^{ - \i \pi(g)}  H(\xi) e^{ \i \pi(g)} )\psi } \leq C(g,f;\xi) \nn{(H_{\f} + \frac{1}{2} \Pf^2 + 1 )\psi} ,  %
\end{equation} 
where 
\begin{align}
C(g,f;\xi) & = |D(f,g)| +  \sqrt 2 \| \omega (f-g) \|_{(0)} \label{eq:diffofconvham}   \\
& + \sum_{j=1}^3  \big\{  ( |d_j(f,g) |+ \sqrt{2} \| K_j(f-g)\|_{(0)} ) \nonumber  \\
& \times (  2 + \sqrt{2}  (  \| e f_{A,j} - K_j f \|_{(0)} +  \| e f_{A,j} - K_j g  \|_{(0)})    + |c_j(f)| + | c_j(g) | )  \big\} \nonumber  \\
& + \sum_{j=1}^3  \sqrt{2} \| K_j^2(f-g)\|_{(0)} ,  \nonumber
\end{align}
with  
\begin{align*}
c_{j}(h)  &:= \xi_j  - 2e {\Re} \mathfrak{s}(h  ,f_{A,j}) - \mathfrak{s}( h , K_j h), \\
  D(f ,g)  & := - 2 \sum_{j=1}^3   e \S_j  {\Re} \mathfrak{s}(f-g,f_{B,j})  +  \mathfrak{s}(f-g,\omega f)  +  \mathfrak{s}(g , \omega ( f - g ) ),  \\ 
 d_j(f,g)  &  := \c_j(f) - c_j(g) =  - 2e {\Re} \mathfrak{s}(f-g  ,f_{A,j}) - \mathfrak{s}( f - g  , K_j f ) - \mathfrak{s}( g , K_j ( f - g ) ) .
\end{align*}

\end{lemma}
\begin{proof} 
First we  take the difference of the operators given in \eqref{eq:gross_transformed_hamiltonian}.
Using   Proposition \ref{prop:gross_transformed_hamiltonian} and  \eqref{eq:defoftm} we find on  $D(\Pf^2 +\Hf)$ 
\begin{align}
&T_0(f;\xi) -e^{ - \i \pi(g)}  H(\xi) e^{ \i \pi(g)} \nonumber  \\
& \quad =  T_0(f;\xi) - T_0(g;\xi)  \nonumber  \\
	& \quad = \frac{1}{2}\sum_{j=1}^3 \left\{ F_{j}(f)^2      - F_{j}(g))^2  \right\}  + \Phi(\omega(f-g))  +  D(f,g) \nonumber  \\
	& \quad = \frac{1}{2}\sum_{j=1}^3 \left\{ F_{j}(f)( F_{j}(f) - F_{j}(g)) + (F_{j}(f)  - F_{j}(g)) F_{j}(g) \right\} \label{eq:tranfoh} 
\\ &\qquad  + \Phi(\omega(f-g))  +  D(f,g)  \label{eq:tranfoh2}
\end{align} 
where we defined  for  $j=1,2,3$,
\begin{align*} 
F_{j}(h)  &:=  - \Pfj{j} + eA_j -  \Phi(K_j h ) + c_j(h)   . 
\end{align*} 
Now the terms in the line \eqref{eq:tranfoh2} are easily 
estimated using  Lemma \ref{lem:basichfest} and 
 contribute  \eqref{eq:diffofconvham}  to the constant.  Now let us 
estimate   the term in line  \eqref{eq:tranfoh}. 
Taking the difference we find 
\begin{align*}
& F_{j}(f)  - F_{j}(g) = -\Phi(K_j(f-g)) + d_j(f,g) .
\end{align*} 
Recall the notation   $ A_j = \Phi( f_{A,j} )$. First we estimate the first  term in the   sum in line   \eqref{eq:tranfoh}.
By the triangle inequality 
\begin{align}
& \|  F_{j}(f)  ( F_{j}(f)  - F_{j}(g) ) \psi \|  \nonumber \\
&  \quad \leq \|   \Pfj{j}   ( F_{j}(f)  - F_{j}(g) )  \psi \|  \label{eq:basicest1}  \\
&  \quad   + \|\Phi(e f_{A,j} - K_j f )    ( F_{j}(f)  - F_{j}(g) )   \psi \| \label{eq:basicest2}  \\
& \quad +   \| c_j(f)  ( F_{j}(f)  - F_{j}(g) ) \psi \| \label{eq:basicest3} .  
 \end{align} 
Now we estimate  the terms on the right hand side.   Using Lemma \ref{lem:basichfest} and Lemma \ref{lem:comdgamma}  we find  
\begin{align}
\text{   \eqref{eq:basicest1}}   & \leq 
  \left\{  |d_j(f,g)| + \sqrt{2} \| K_j(f-g)\|_{(0)} \right\}   \| (\Hf + 1 )^{1/2} P_{f,j} \psi \|  \nonumber \\
& \quad    + \sqrt{2} \| K_j^2(f-g)\|_{(0)} \| (\Hf+1)^{1/2} \psi \|  . \label{eq:extratermin} 
\end{align}
Again by   Lemma \ref{lem:basichfest}  we have 
\begin{align*} 
\text{   \eqref{eq:basicest2}}     & \leq  ( |d_j(f,g) |+ \sqrt{2} \| K_j(f-g)\|_{(0)} ) \sqrt{2}   \| e f_{A,j} - K_j g \|_{(0)}  \| ( \Hf + 1 ) \psi \| ,
\\
\text{   \eqref{eq:basicest3}}   & \leq (|d_j(g,f)| + \sqrt{2} \| K_j(f-g) \|_{(0)} )  |c_j(f) | \| (  \Hf + 1  )^{1/2}  \psi  \|.
 \end{align*} 
The second  term in the sum of    \eqref{eq:tranfoh} is estimated analogously with $g$ and $f$ interchanged, with the only difference that we do not need 
the   term in line   \eqref{eq:extratermin}.
The lemma now follows by collecting estimates and observing  that $d_j(g,f) = - d_j(f,g)$. 
\end{proof}

Now \cref{th:renormalized self-adjoint}  will follow as an immediate consequence of the following proposition. 
Furthermore, we will show below that it also  implies \cref{thm:existoftrafoham}. 
\begin{proposition}\label{convham} 
Let $E$ be differentiable in $\xi$ and $|\nabla E(\xi)| < 1$. 
 If  $m > 0$ and $|\nabla E_m(\xi)| < 1$, then  $h_{m,\xi }\in \hh$.  Furthermore, 
assume $(m_j)_{j \in \N}$ is a sequence of positive  numbers converging to zero  and 
 $\nabla E_{m_j}(\xi) \to \nabla E(\xi)$.  Then 
\begin{equation} 
\label{eq:HHm_estm}
 \left(e^{ - \i \pi(h_{m_j,\xi})}  H(\xi) e^{ \i \pi(h_{m_j,\xi})} -   T_0(h_{0,\xi};\xi) \right) (\Pf^2 + \Hf + 1 )^{-1} \to 0  , \quad ( j \to  \infty ) .
\end{equation}
\end{proposition} 
\begin{proof} For $m > 0$ and $|\nabla E_m(\xi)|<1$  we have $h_{m,\xi} \in \hh$, since $\omega_m(k) - k \nabla E_m(\xi) >  \omega_m(k) (1-  |\nabla E_m(\xi)|) $. 
Statement  \eqref{eq:HHm_estm} follows by inserting $h_{0,\xi}$  for $f$ and  $h_{m,\xi}$, $m>0$, for $g$  
in  Lemma \ref{eq:thmconren2} and by observing  that  the constant  $ C(h_{m,\xi},h_{0,\xi} ; \xi)$
given in \eqref{eq:diffofconvham}
 tends to 
zero as $m\downarrow 0$. This can be seen by looking at the explicit expressions
and using  dominated convergence. 
\end{proof}
\begin{proof}[Proof of \cref{th:renormalized self-adjoint}]
The operators  $e^{ - \i \pi(h_{m,\xi})}  H(\xi) e^{ \i \pi(h_{m,\xi})}$ 
are self-adjoint on $D(\Pf^2 + \Hf)$  by \cref{prop:gross_transformed_hamiltonian}.
On the other hand    $ T_0(h_{0,\xi};\xi)$ is  $\Pf^2 + \Hf$ bounded by Lemma  \ref{lem:asslem}.
So the  self-adjointness of  $T_0(h_{0,\xi};\xi)$ follows 
in view of  %
 Kato-Rellich, and the following estimate 
\begin{align} \label{saofmain} 
& \| ( e^{ - \i \pi(h_{m,\xi})}  H(\xi) e^{ \i \pi(h_{m,\xi})} - T_0(h_{0,\xi};\xi) )(e^{ - \i \pi(h_{m,\xi})}  H(\xi) e^{ \i \pi(h_{m,\xi})} + \i )^{-1}  \|  \\
&  \leq \| ( e^{ - \i \pi(h_{m,\xi})}  H(\xi) e^{ \i \pi(h_{m,\xi})} - T_0(h_{0,\xi};\xi) ) (\Pf^2 + \Hf+1)^{-1} \| \nonumber   \\
& \quad  \times \|( \Pf^2 + \Hf+1)  (e^{ - \i \pi(h_{m,\xi})}  H(\xi) e^{  \i \pi(h_{m,\xi})}  + \i )^{-1}  \| \nonumber     \\
&  \leq \| ( e^{ - \i \pi(h_{m,\xi})}  H(\xi) e^{ \i \pi(h_{m,\xi})} - T_0(h_{0,\xi};\xi) ) ( \Pf^2 + \Hf+1)^{-1} \| \nonumber    \\
& \quad \times \|e^{  \i \pi(h_{m,\xi})} (\Pf^2 + \Hf+1)  e^{ - \i \pi(h_{m,\xi})}  ( H(\xi) + \i )^{-1}  \| \nonumber    \\
&  \leq \| ( e^{  \i \pi(h_{m,\xi})}  H(\xi) e^{ \i \pi(h_{m,\xi})} - T_0(h_{0,\xi};\xi) ) (\Pf^2 + \Hf+1)^{-1} \| \nonumber    \\
& \quad \times \| (\Pf^2 + \Hf+1)  e^{  - \i \pi(h_{m,\xi})}  (\Pf^2 + \Hf+1)^{-1}  \|  \| (\Pf^2 + \Hf+1)   ( H(\xi) + \i )^{-1}  \| \nonumber  .
\end{align}
Now the last term on the right hand side of  \eqref{saofmain}  is bounded by  Theorem \ref{hyp:op0},    the second term on the right hand side is bounded uniformly in $m \geq 0$ 
in view of   the estimate \eqref{eq:trafoofpfhfm}  in  Lemma \ref{eq:lemdomain},  and the first term on the right hand side of  \eqref{saofmain} tends to zero as $m \downarrow 0$ by \cref{convham}. 
\end{proof}

\begin{proof}[Proof of Theorem \ref{thm:existoftrafoham}]
By \cref{convham} (a) the resolvent $( T_0(h_{0,\xi};\xi) + \i )^{-1}$ is well-defined and we obtain
\begin{align} 
& \| ( e^{ - \i \pi(h_{m,\xi})}  H(\xi) e^{ \i \pi(h_{m,\xi})} - T_0(h_{0,\xi};\xi)     )( T_0(h_{0,\xi};\xi) + \i )^{-1} \|   \nonumber \\
& \leq
 \| (e^{ - \i \pi(h_{m,\xi})}  H(\xi) e^{ \i \pi(h_{m,\xi})} - T_0(h_{0,\xi};\xi) ) (H_{\f} + \Pf^2 + 1 )^{-1} \|  \nonumber \\
& 
\quad \times  \|   (H_{\f} +  \Pf^2 + 1 )  (T_0(h_{0,\xi};\xi) + \i )^{-1} \|   .  \label{normresconv} 
\end{align}  
Now the right hand side  of  \eqref{normresconv}   tends to zero as $m\downarrow 0$  by  (b) of  Proposition \ref{convham}. This implies norm resolvent convergence.
We recall that for a sequence of self-adjoint operators $(A_n)_{n \in \N}$ to converge in norm resolvent sense to a selfadjoint operator $A$, it suffices to show
 that  $\| (  A_n - A )(A+ \i)^{-1} \| \to 0$  as $n \to \infty$, which can be seen from  \cite[Theorem VIII.19]{rs1} and the first resolvent identity. 
\end{proof}

In the following proposition we will establish  that $U_{m_j}(\xi) \psi_{m_j}(\xi)$ is a minimizing sequence of $\widehat{H}(\xi)$.

\begin{proposition}  \label{convofenergyexp} 
Suppose $|\xi | < 1$.  
  Furthermore, 
assume $(m_j)$ is a sequence of positive numbers converging to zero  and 
 $\nabla E_{m_j}(\xi) \to \nabla E(\xi)$.  Then
\[
0 \leq \inn{ U_{m_j}(\xi) \psi_{m_j}(\xi) , ( \widehat{H}(\xi) - E(\xi)  ) U_{m_j}(\xi) \psi_{m_j}(\xi) } \to 0 
\]
in the limit $j  \to \infty$. 
\end{proposition}

\begin{proof} 
We recall  that  the assumption $|\xi| < 1$ guarantees by  Lemma   \ref{eq:convex}
that always $|\nabla E_{m}(\xi) | < 1$ and that, for $m > 0$, $U_m(\xi)$ is a well defined unitary operator.
 Therefore we know that  $h_{m,\xi} \in \hh$. 
For the proof we define 
$$
\widehat{H}_m(\xi) := U_m(\xi) H(\xi) U_m(\xi)^* . 
$$
Then we have by definition
\begin{align} 
& \sc{U_m(\xi) \psi_{m}, (\widehat{H}(\xi) - E(\xi))U_{m_j} \psi_{m_j}} \nonumber \\
&=  \sc{U_m(\xi) \psi_{m}, (\widehat{H}(\xi) - \widehat{H}_{m}(\xi) + 
 \widehat{H}_{m}(\xi)  -  E(\xi))U_m(\xi) \psi_{m}} \nonumber \\
&=   \sc{U_m(\xi) \psi_{m}, (\widehat{H}(\xi) - \widehat{H}_{m}(\xi) )  U_m(\xi) \psi_{m}}
 + \sc{\psi_{m} ,  ( {H}(\xi)  -  E(\xi)) \psi_{m}}  . \label{eq:last Um}
\end{align} 
 Now we use Proposition  \ref{convham}  to control the first term of \eqref{eq:last Um} along the sequence $(m_j)$
\begin{align*} 
 ( \widehat{H}(\xi) - \widehat{H}_{m_j}(\xi))^*  (  \widehat{H}(\xi) - \widehat{H}_{m_j}(\xi) )  & \leq C_j^2 ( \Hf + \Pf^2 + 1 )^2 , 
\end{align*} 
with   $C_j \to 0$ as $j  \to  \infty$.
Taking the square root yields 
\begin{align*} 
  \widehat{H}(\xi) - \widehat{H}_{m_j}(\xi) \leq |   \widehat{H}(\xi) - \widehat{H}_{m_j}(\xi) |  & \leq  | C_j| ( \Hf + \Pf^2 + 1 ) .
\end{align*} 
Inserting this above we find  %
\begin{align*} 
\sc{U_{m_j}(\xi) \psi_{m_j}, (\widehat{H}(\xi) - \widehat{H}_{m_j}(\xi) )  U_{m_j}(\xi) \psi_{m_j}}
\leq  |C_j|
\langle  U_{m_j}(\xi)\psi_{m_j}  , (H_{\f} + \Pf^2 + 1  )U_{m_j}(\xi)\psi_{m_j} \rangle   . 
\end{align*} 
The second term in \eqref{eq:last Um} can be bounded by means of $H(\xi) \leq H_{m}(\xi)$  
\begin{align*}
\sc{\psi_{m} ,  ( {H}(\xi)  -  E(\xi)) \psi_{m}}  \leq \sc{\psi_{m} ,  ( {H}_{m}(\xi)  -  E(\xi)) \psi_{m}}  
=    E_m(\xi) - E(\xi)  .
\end{align*}
Now, as $H(\xi)$ is closed on $D(\Pf^2 + \Hf)$ we find with Lemma \ref{eq:lemdomain}  
for some constant $C$ independent of $m \geq 0$ 
\begin{align*}
\langle  U_m(\xi)\psi_{m}  , (H_{\f} + \Pf^2 + 1  )U_m(\xi) \psi_{m} \rangle & 
\leq C\langle  \psi_{m}  , (H_{\f} + \Pf^2 + 1  ) \psi_{m} \rangle \\
& \leq C\langle  \psi_{m}  ,  ( H(\xi)  + 1  ) \psi_{m} \rangle \\ 
& \leq C\langle  \psi_{m}  ,  ( H_{m}(\xi)  + 1  ) \psi_{m} \rangle  = C ( E_{m}(\xi) + 1 ) ,
\end{align*} 
which is in fact bounded by Proposition  \ref{eq:propenegyconv}.
Collecting estimates and using  that $E_m(\xi) \to E(\xi)$ by Proposition  \ref{eq:propenegyconv}, the claim follows. 
\end{proof}

 \section{Perturbation Theory} 
\label{sec:pertheory} 

In this section we use perturbation theory to obtain resolvent bounds.
More precisely, we use second order perturbation theory to determine the 
second derivative  of the energy in terms of the resolvent. Then using bounds on the 
second derivative of the energy, we find bounds on the resolvent. 
These bounds will be used to  obtain infrared bounds for the ground states of 
the massive Hamiltonians. 

If $|\xi| \leq 1$, then we know by \cref{hyp:op1} that $E_m(\xi)$ is an  
eigenvalue of $H_m(\xi)$ isolated from the essential spectrum. Let $P_m(\xi)$ 
denote the projection onto the kernel of $H_m(\xi)$, which is finite dimensional. 
We recall the definition of $v(\xi)$ in \eqref{defofv}. 

\begin{lemma} \label{lem:per} Let $e \in \R$ und $m > 0$, and suppose 
  \eqref{eq:eineq} holds. Let $|\xi| \leq 1$ and suppose that 
 all first order partial derivatives of $E_m(\cdot)$  exist  at $\xi$. Then the following holds.
\begin{enumerate}[(a)]  
\item $
P_m v(\xi) P_m = \nabla_\xi E_m (\xi) P_m$.
\item   \label{lem:1}  $E_m( \cdot)$ is twice partially differentiable  at $\xi$  and 
for  $j=1,2,3$, 
\begin{align*} 
& \partial_j^2 E_m(\xi)  \\
& \ \leq 1  - 2 P_m(\xi) ( v(\xi) - \nabla E_m(\xi) )_j \frac{1}{H_m(\xi) - E_m(\xi)}   ( v(\xi) - \nabla E_m(\xi) )_j  P_m(\xi) 
\end{align*} 
as inequality in the sense of  operators  on $\ran P_m(\xi)$. 
\end{enumerate} 
\end{lemma}
\begin{rem}
We note that by (a)  the vectors in $\ran ( v(\xi) - \nabla E_m(\xi) )_j  P_m(\xi)$ are  orthogonal to $\ran P_m(\xi)$ 
and hence the resolvent in (b) is well defined.  
\end{rem} 

\begin{proof}
For the proof we use analytic perturbation theory. For details we refer the reader to  \cite{kato,ressim:ana}. 
On $\C^3$ the operator valued function $\zeta \mapsto
H_m(\xi + \zeta) := H_m(\xi)  + \zeta \cdot v(\xi)  +
\frac{1}{2}\zeta^2$ is an analytic family of type (A) in each component.
By  Theorem  \ref{hyp:op1}   we
know that  $E_m(\xi)$ is an eigenvalue isolated from the essential spectrum, and thus with finite 
multiplicity     $n(\xi) = \dim \ran P_m(\xi) $. 
\\
 By degenerate perturbation theory, cf. \cite[Theorem XII.13]{ressim:ana}, we know that there exist $n(\xi)$ 
complex analytic functions $e_s  $, $s=1,\ldots,n(\xi)$, in a neighborhood of zero
such that
 \begin{equation} \label{eq:perteq}
e_s(0) = E_m(\xi), \quad s = 1 , \ldots, n(\xi) ,
\end{equation} 
 and $e_s(\zeta)$ is an eigenvalue of $H_m(\xi + \zeta)$. 
The functions $e_s$, $s=1,\ldots,n(\xi)$,  are real for real $\zeta$, and 
$$
\partial_j e_s(0) , \quad s = 1, \ldots, n(\zeta), 
$$
are the eigenvalues of  $$P_m(\xi) v_j(\xi) P_m(\xi), $$ 
which can be seen from a power series expansion and comparison of coefficients.
Since $x \mapsto E_m(\xi + x) = \inf_{s=1,\ldots,n}  e_s(x)$
 is 
differentiable   at $x=0$, by assumption,   and  $e_s(0) = E_m(\xi)$, it can be seen from a first order Taylor expansion that all derivatives  must be equal at the origin and equal $\partial_j E(\xi)$, i.e., 
\begin{equation} 
\partial_j  e_s( 0) = \partial_j E(\xi)  ,  \quad s=1,\ldots,n(\xi) . \label{derallsame}  
\end{equation} 
Thus, $\nabla E_m(\xi) P_m(\xi) = P_m(\xi) v(\xi) P_m(\xi)$. This shows (a). 
Now  by second order degenerate perturbation theory, we find 
$$
\partial_j^2 e_s(0) , \quad s = 1 ,\ldots, n(\xi) ,   
$$
are the eigenvalues of 
$$
P_m(\xi) - 2 P_m(\xi) ( v(\xi) - \nabla E(\xi) ) (H_m(\xi) - E_m(\xi))^{-1}  ( v(\xi) - \nabla E(\xi) )  P_m(\xi),
$$
which can again be seen from a power series expansion. 
By \eqref{derallsame} and \eqref{eq:perteq}   and by the definition of  $E_m(\xi)$ as an  infimum,
it can be seen from a second order Taylor expansion that $E_m$ is twice differentiable at $\xi$ and  that 
$$
\partial_j^2 E_m(\xi) = \inf_s \partial_j^2 e_s(0) .  
$$
This shows (b). 
\end{proof}

\section{Infrared Bounds}

\label{sec:infrabounds}

In this section we derive infrared bounds. These will then be 
used  in the next section to show that infrared regularized  ground states for a fixed momentum
lie in a compact set. 
For $\xi \in \R^3$ and $m \geq 0$ we define the following quantity 
$$
\Delta_m(\xi) = \inf_{k \in \R^3} \{ E_m(\xi - k) - E_m(\xi) + \omega_m(k) \} ,
$$
which we shall use to estimate the resolvent of the Hamiltonian.
We  will use the notation 
\muu {
v := v(0) = - \Pf + e A .
} 
The following proposition tells us when this resolvent  defined in  \eqref{defofR} 
is well defined. 

\begin{prop} \label{prop:Deltaineq}  Let $e \in \R$ and $m > 0$, and suppose the energy Inequality \eqref{eq:eineq}
holds. Then the following holds.
\begin{itemize}
\item[(a)]   We have  $\Delta_m(\xi) > 0$, whenever $|\xi| \leq 1$. 
\item[(b)]   For  $|\xi| \leq 1$ and  all $k \in \R^3$ we have 
$$ H_m(\xi-k) + \omega_m(k) - E_m(\xi)  > 0 . $$
\item[(c)]   For  $|\xi|  <  1$  and all $k \in \R^3$ 
\begin{equation}   \label{prop:Deltaineqc}
\nn{R_{m,\xi}(k)} \leq \frac{1}{1-|\xi|}\frac{1}{\abs k}.
\end{equation} 
\end{itemize} 
\end{prop} 
\begin{proof} First we show that  the function $\xi \mapsto E_m(\xi)$ satisfies  the assumptions (i)--(iii) of  Proposition \ref{prop:convexlms} in the Appendix. 
Property (i) is simply the energy inequality   \eqref{eq:eineq}.  From Lemma  \ref{eq:convex} we know on the one hand  that  $t_m(\xi) = \frac{1}{2} \xi^2 - E_m(\xi)$ 
is convex, i.e., (iii) and on the other hand that   $E_m(-\xi) = E_m(\xi)$, which implies   $t_m(-\xi) = t_m(\xi)$.  
Thus by convexity  $t_m(0) \leq t_m(\xi)$, which implies (ii).     
It therefore follows from  Proposition \ref{prop:convexlms} 
\begin{equation} \label{eq:ineqnow} 
E_m(\xi - k ) - E_m(\xi) \geq \left\{ \begin{array}{ll} - |k| | \xi | + \frac{1}{2} k^2  & \text{,  if } |k| \leq |\xi | , \\ -\frac{1}{2} \xi^2  & \text{, if } |k| \geq |\xi| .\end{array} \right. 
\end{equation} 
(a)  We find using  $\omega(k) > |k|$ in \eqref{eq:ineqnow} 
\begin{align*}
E_m(\xi - k ) - E_m(\xi) + \omega(k) 
  & >  \left\{ \begin{array}{ll} - |k| | \xi | + |k|   & \text{,  if } |k| \leq |\xi | , \\ -\frac{1}{2} \xi^2 + |\xi|  & \text{, if } |k| \geq |\xi| , \end{array} \right.  \\
 &  \geq  0 , 
\end{align*} 
provided  $|\xi | \leq  1$. \\
 (b) From (a)  we have  
$H_m(\xi - k) + \omega_m(k) - E_m(\xi) \geq E_m(\xi-k) + \omega_m(k) - E_m(\xi) \geq \Delta_m(\xi) > 0. $\\
(c)  From \eqref{eq:ineqnow} we find 
\begin{align*} 
H_m(\xi -k) + \omega_m(k) - E_m(\xi) &\geq E_m(\xi - k) - E_m(\xi) + |k|  \\
&\geq \begin{cases} |k|(1-|\xi|)  &\text{, if } \abs k \leq \abs \xi, \\
				    |k| - |\xi|^2/2 &\text{, if }  \abs k \geq \abs \xi. \end{cases}
\end{align*} 
Observing that $|k| - |\xi|^2/2 \geq |k| - \frac{|\xi||k|}{2} \geq (1-|\xi|)|k|$ if $\abs k \geq \abs \xi$ shows the claim.
\end{proof}

Now we use the following relation, which we shall refer to as the pull-through resolvent identity.
To formulate it we shall make use of the physics notation of the pointwise annihilation operator.
One defines for $(\lambda,k)\in \Z_2\times \R^3$  and $\psi \in \FF$
\begin{align} \label{eq:defofannihdir}
[a_\lambda(k) \psi_{(n)}](\lam_1, {k}_1, \ldots , \lam_n, {k}_n)  & = \sqrt{n+1} \psi_{(n+1)}(\lambda, k , \lam_1, {k}_1, \ldots , \lam_n, {k}_n)  , \, n \in \N_0 .
\end{align}
Note that by the Fubini-Tonelli theorem  $ a_\lambda(k) \psi_{(n+1)}  \in \hh^{(n)}$ for almost every $k$.
If $\psi \in D(\Hf)$  and $f \in \hh$ with $\| f \|_{(0)} < \infty$, then  a straightforward calculation using 
the definitions shows that 
in the sense of  weak integrals 
\begin{align} \label{eq:defofafint} 
a(f) \psi   & = \sum_{\lambda=1,2} \int   d{k}   \overline{f(\lam, {k}) } a_\lam({k}) \psi   . 
\end{align} 
In fact,  the integral on the right hand side  \eqref{eq:defofafint} exists as a Bochner integral. 
This can be seen from the following estimate  
\begin{align*}
\sum_{\lambda=1,2} \int  \| \overline{f(\lambda,k) } a_\lambda(k)   \psi \| dk & = \sum_{\lambda=1,2} \int dk |  \overline{f(\lambda, k) } \omega(k)^{-1/2} |  \| \omega(k)^{1/2} a_\lambda(k)   \psi \| dk \\
& \leq \| \omega^{-1/2} f \| \left( \sum_{\lambda=1,2}  \int dk \inn{ a_\lambda(k) \psi ,  \omega(k) a_\lambda(k) \psi }  \right)^{1/2}  \\
&  =    \| \omega^{-1/2} f \| \| \Hf^{1/2 } \psi \| .
\end{align*}

\begin{lemma}[Pull-through resolvent identity]
\label{lem:pullthrough}
Let $e \in \R$ and $m > 0$, and suppose the energy Inequality \eqref{eq:eineq}
holds. Suppose   $|\xi|  <  1$. Then for all $k \in \R^3$ and $\lambda \in \{1,2\}$ we have  
\begin{align} 
a_\lambda(k) \psi_m(\xi) =  \frac{e \rho(k)}{\sqrt{2 \abs k }}  R_{m,\xi}(k)\left(- \epsilon_\lambda(k) v(\xi) +  S \cdot (\i k \wedge \epsilon_\lambda(k)) \right) \psi_m(\xi) . \label{eq:pullthroughrel}
\end{align} 
\end{lemma}
\begin{proof}[Sketch of proof]
For details we refer to \cite[Lemma 6.1]{chenfroehlich} and \cite[Lemma 7]{haslerherbst1}.
As  $\psi_m \in D(\Hf)$, we have
$$
\sum_{n=0}^\infty \sum_{\lambda} \int |k| \nn{ ( a_{\lambda}(k)
\psi_m )_{(n)} }^2 dk  = \langle \psi_m , \Hf \psi_m \rangle < \infty \; ,
$$
which implies $\nn{a_\lambda(k) \psi_m} < \infty$ almost everywhere.

The pull-through formula yields, e.g. \cite[Lemma A.1]{BachFroehlichSigal.1998a}, for a.e. $k \in \IR^3$, $\lambda \in \{\ps \}$,  in the sense of measurable fucntions 
\muu
{
a_\lambda(k) H_m(\xi) \psi &= \left(\frac{1}{2}(\xi - k - \Pf + eA)^2 + e S \cdot  B + \Hfm + \omega_m(k) \right) a_\lambda(k) \psi  \\ &\qquad+ ( e f_A(k) \cdot  v(\xi) + e S \cdot f_B(k)  ) \psi 
}
for any $\psi \in \FF$.
Therefore,
\muu
{
a_\lambda(k) E_m(\xi) \psi_m(\xi) &= a_\lambda(k) H_m(\xi) \psi_m(\xi) \\
&= ((H_m(\xi - k) + \omega_m(k)) a_\lambda(k)  + e f_A(k) \cdot  v(\xi) + e S \cdot  f_B(k))\psi_m(\xi),
}
or in other words,
\muu
{
((H_m(\xi - k) - E_m(\xi) + \omega_m(k)) a_\lambda(k) \psi_m(\xi) = - (ef_A(k) \cdot v(\xi) + e S \cdot f_B(k) )\psi_m(\xi).
}
Multiplying with $R_{m,\xi}(k)$ yields the desired formula.
\end{proof}

The next lemma is needed to estimate the resolvent occurring in \eqref{eq:pullthroughrel}. 
It uses the   estimate obtained by  second order perturbation theory stated in  Lemma \ref{lem:per}.

\begin{lemma} \label{lem:estimate1} There exists a constant $C$  such that the following holds. 
Let $e \in \R$ and $m > 0$, and suppose the energy Inequality \eqref{eq:eineq}
holds.  Then 
for all   $\xi \in \R^3$ with   $|\xi|  <  1$,  such  that $\nabla E_m(\xi)$  exists, and  all $k \in \R^3$ and $i =1,2,3$,  we have 
\begin{enumerate}[label=(\alph*)]
\item $\nn{ R_{m,\xi}(k) (v_i(\xi)- \partial_i E_m(\xi)) \psi_m(\xi) }  \leq  \frac{ C   |\frac{1}{2} - \partial_i^2 E_m(\xi)|^{1/2}    (1+|E_m(\xi)|)^{1/2}   }{ 1-|\xi| }
(|k|^{-1/2} + 1)$, \label{lem:estimate111}
\item
 \label{en:estimate1_en2}   $\nn{R_{m,\xi}(k) v_i(\xi) } \leq C (1-|\xi|)^{-1}    (1+|E_m(\xi)|)^{1/2}
 (|k|^{-1} + 1)$.
\end{enumerate}
\end{lemma}

\newcommand{\sqrtHmEm}{(H_m - E_m)^{\frac{1}{2}}}
\begin{proof}
To simplify the notation we  drop the $\xi$ label in $\psi_m(\xi)$,  $R_{m,\xi}(k)$, and $v(\xi)$, and we  write 
\begin{align*} 
 h_m   := H_m(\xi), \qquad  e_m  := E_m(\xi)  .
\end{align*} 
(a) We start with the product inequality
\begin{eqnarray} \label{eq:prod}
\| R_m(k) v_{j} \psi_m \| \leq  \| R_m(k)
( h_m - e_m )^{1/2} \| \| ( h_m - e_m )^{-1/2} (  v_{j} - \partial_j E_m(\xi)) \psi_m \| .
\end{eqnarray}
By Lemma \ref{lem:per} the second factor on the right hand side can be estimated using
\[
\| (h_m - e_m)^{-1/2}  (v_j -\partial_j E_m(\xi) ) \psi_m \| \leq \left| \frac{1}{2} - \partial_j^2 E_m(\xi) \right|^{1/2}  \; .
\]
It remains to estimate the first factor in \eqref{eq:prod}.
First we use the trivial identity
\muu
{
h_m - e_m &= \frac{1}{2} (v-k)^2 + H_{f,m} - e_m + \frac{1}{2}k^2 + (v-k)\cdot k .
}
Estimating the last term using
\[
 (v-k)\cdot k \leq \frac{1}{2} |k| + \frac{1}{2} |k| (v-k)^2 ,
\]
we find
with
$
 \frac{1}{2} (v-k)^2 \leq H_m(\xi -k)
$ that
\begin{align} \label{eq:hmemineq2}
h_m - e_m
&\leq (1 + | k |) (H_m(\xi -k) + \omega_m(k) - e_m) + \frac{1}{2}(|k| +k^2) + | k | e_m .
\end{align}
Now  multiplying this inequality on both sides  with the self-adjoint operator $R_m(k)$ we obtain
\muu
{
R_m(k) (h_m - e_m) R_m(k) \leq  (1+ | k |) R_m(k) + \left( \frac{1}{2} ( | k |  + k^2) +  |  k | e_m \right) R_m(k)^2.
}
Using this, we estimate
\begin{align}
& \| R_m(k) (h_m - e_m)^{1/2}  \|^2 \\
&  \qquad \leq  \|  (h_m - e_m)^{1/2} R_m(k) \|^2 \nonumber \\
& \qquad = \sup_{\|\phi\| = 1} \inn{\phi,R_m(k)(h_m - e_m) R_m(k) \phi}  \nonumber   \\
& \qquad \leq (1 + | k |) \| R_m(k)\| +  \left( \frac{1}{2} ( |k | + k^2) + |k | e_m \right) \| R_m(k)\|^2.
\label{eq:last_term}
\end{align}
Now using Inequality   \eqref{prop:Deltaineq}
this   implies the bound stated in (a). \\
(b)  Using that $v^2 \leq h_m$ we see from  \eqref{eq:hmemineq2} that
\begin{align*}
v^2
&\leq (1 + | k |) (H_m(\xi -k) + \omega_m(k) - e_m) + \frac{1}{2}(|k| +k^2) + (1 + | k |) e_m .
\end{align*}
This implies 
\begin{align}
\nonumber   \left\|R_m(k) v_i \right\|^2 \leq (1 + | k |) \| R_m(k) \| +  \left( \frac{1}{2} (  |k | + k^2) + ( | k | +1) e_m \right) \| R_m(k)\|^2 .
\end{align}
Now using  Inequality   \eqref{prop:Deltaineq}  this implies the bound in  (b).
\end{proof}

We note that  the assumption $|\xi| < 1$ guarantees by  Lemma   \ref{eq:convex}
that $|\nabla E_{m}(\xi) | < 1$ always holds, provided the derivative exists. Therefore we know that  $h_{m,\xi} \in \hh$ and the following two expressions are well defined and positive, 
\begin{align}
D_{m,\xi}(k) &  := (\omega_m(k)  -  k \cdot \nabla E_m(\xi) )^{-1}  , \\ 
\widehat{D}_{m,\xi}(k) &  :=  (1  -  \omega_m(k)^{-1} k \cdot \nabla E_m(\xi) )^{-1} .  %
\end{align} 
To write the following  estimates in compact notation we  define
\begin{equation} \label{eq:deffoff1}
F_{m,\xi}    :=  \sum_{i=1}^3 \left(  1 +   \abs{ \frac{1}{2} - \partial_i^2 E_m(\xi)}^{1/2}  \right)   (1+|E_m(\xi)|)^{1/2}.
\end{equation} 
Recall that if $\nabla E_m(\xi)$ exists, then also the second derivative, cf. \cref{lem:per}.

\begin{prop}      There exists a constant $C$ such that the following holds.
\label{prop:a_est_2}
Let $e \in \R$ and $m > 0$, and suppose the energy Inequality \eqref{eq:eineq}
holds.   Let    $\xi \in \R^3$ with   $|\xi|  <  1$,  such  that $\nabla E_m(\xi)$  exists.
Then for  a.e.   $k \in \IR^3$ 
\[
\nn{a_\lambda(k) U_m(\xi) \psi_m(\xi)} \leq \frac{C |e \rho(k)|  F_{m,\xi} }{1-|\xi|} (  \widehat{D}_{m,\xi}(k)  + 1 )      (|k|^{-1} + |k|^{-1/2} ).
\]
\end{prop}
\begin{proof}   Using the convergence of \eqref{eq:defofafint} as a Bochner integral and 
  Lemma \ref{lem:bratt} we find   for any $f \in \hh$  with $\| f \|_{(0)} < \infty$ and $\varphi \in \FF_{\fin}(\hh)$ 
that 
\begin{align}
& \sum_{\lambda=1,2} \int dk \overline{f(\lambda,k)} \inn{ \varphi , U_m(\xi)^* a_\lambda(k) U_m(\xi) \psi_m(\xi) }  \label{eq:conjann-1}   \\
& = \inn{ \varphi , U_m(\xi)^* a(f) U_m(\xi) \psi_m(\xi)}  \nonumber   \\
 & = \sum_{\lambda=1,2} \int dk  \overline{f(\lambda,k)} \inn{ \varphi ,  a_\lambda(k) \psi_m(\xi) }   \nonumber  \\
& \quad  + \sum_{\lambda=1,2} \int dk \overline{f(\lambda,k)}  e \epsilon_\lambda(k) \cdot \nabla_\xi E_m(\xi)  \frac{\rho(k)}{\sqrt{2 \abs{k}}} \frac{1}{\omega_m(k) - k\nabla_\xi \cdot E_m(\xi)} \inn{ \varphi ,  \psi_m(\xi)}  . 
\nonumber 
\end{align} 
Since $f \in \hh$  with $\| f \|_{(0)} < \infty$ and $\varphi \in \FF_{\fin}(\hh)$  are arbitrary it follows
 from \eqref{eq:conjann-1}  and density of $\FF_{\fin}(\hh)$
 that for a.e. $k \in \R^3$  
\begin{align}
& U_m(\xi)^* a_\lambda(k) U_m(\xi) \psi_m(\xi) \nonumber   \\
 & =  a_\lambda(k) \psi_m(\xi)  + e \epsilon_\lambda(k) \nabla_\xi E_m(\xi)  \frac{\rho(k)}{\sqrt{2 \abs{k}}} \frac{1}{\omega_m(k) - k\nabla_\xi E_m(\xi)} \psi_m(\xi) . \label{eq:conjann0} 
\end{align} 
The goal is to show that as   $k \to 0$ the leading order contribution of the first 
term in \eqref{eq:conjann0} cancels the second term. %
To show this we use the  pull-through resolvent identity in \cref{lem:pullthrough} and can write 
\begin{equation} 
a_\lambda(k) \psi_m(\xi) = (\mathrm{I}) + (\mathrm{II}) , 
\end{equation} 
where 
\begin{align*}
(\mathrm{I}) &:=  R_{m,\xi}(k)  \frac{e \rho(k)}{\sqrt{ 2 \abs k } } S \cdot ( \i k \wedge \epsilon_\lambda(k))  \psi_m(\xi), \\
(\mathrm{II}) &:= -R_{m,\xi}(k) \frac{e \rho(k)}{\sqrt{2 \abs k }} \epsilon_\lambda(k) \cdot  v(\xi)   \psi_m(\xi).
\end{align*}
 
The first term,  (I), can be estimated using  %
  \eqref{prop:Deltaineqc} 
\muu
{
\nn{(\mathrm{I})} \leq \frac{|e\rho(k)|}{1-|\xi|} \frac{1}{ \sqrt{2 \abs{k}}}.
}
Furthermore, we write the second term 
$$
(\mathrm{II}) = (\mathrm{II})_1 + (\mathrm{II})_2 $$
 by dividing it into the following two parts 
\muu
{
(\mathrm{II})_1 &:= - R_{m,\xi}(k) \frac{e \rho(k)}{\sqrt{2 \abs k}} \epsilon_\lambda(k) \cdot ( v(\xi)-\nabla E_m(\xi))  \psi_m(\xi), \\
(\mathrm{II})_2 &:= - R_{m,\xi}(k)  \frac{e \rho(k)}{\sqrt{ 2 \abs k }} \epsilon_\lambda(k) \cdot \nabla E_m(\xi)  \psi_m(\xi).
}
We have 
\[
\nn{(\mathrm{II})_1} \leq C   \frac{|e \rho(k)|}{\sqrt{2 \abs k}}  F_{m,\xi}  (1-|\xi|)^{-1} (|k|^{-1/2} + 1)
\]
by \cref{lem:estimate1}. 
To estimate $(\mathrm{II})_2$, we proceed as follows. Similarly  to  \cite{haslerherbst1} we introduce the operator
\begin{align*} 
R^{(0)}_{m,\xi}(k)  & :=  (
  H_m(\xi ) + \omega_m(k) - E_m(\xi)   )^{-1}  .  
\end{align*}
Using the second resolvent identity, we find 
\begin{align}
R_{m,\xi}(k) \psi_m(\xi)   %
&=   R_{m,\xi}^{(0)}(k)\psi_m + R_{m,\xi}(k)  ( - \frac{1}{2}k^2 + k \cdot v(\xi) ) R_{m,\xi}^{(0)}(k) \psi_m(\xi)  \nonumber \\
&= \omega_m(k)^{-1}  \psi_m +   \omega_m(k)^{-1}  R_{m,\xi}(k)  (  k \cdot v(\xi) - \frac{1}{2}k^2 ) \psi_m(\xi)   \nonumber \\
&=\omega_m(k)^{-1} \left(   \psi_m +   k \cdot  \nabla E_m(\xi) R_{m,\xi}(k)   \psi_m(\xi)  +  ({\mathrm{III}})_1 + ({\mathrm{III}})_2  \right) , \label{bringleft44}
\end{align} 
where we defined 
\begin{align*}
({\mathrm{III}})_1 & := 
  -      \frac{1}{2}k^2 R_{m,\xi}(k)  \psi_m(\xi),  \\
(\mathrm{III})_2 &  := 
  R_{m,\xi}(k)   k\cdot  ( v(\xi)  - \nabla E(\xi)) \psi_m (\xi) .
\end{align*} 
Hence, multiplying out  \eqref{bringleft44} and bringing the  term  $\omega_m(k)^{-1}   k \cdot  \nabla E_m(\xi) R_{m,\xi}(k)   \psi_m(\xi)$ to the left, we arrive at 
\muun
{
\label{eq:Rmxi_formula}
R_{m,\xi}(k) \psi_m(\xi)  =  D_{m,\xi}(k)  ( \psi_m(\xi)  
+   (\mathrm{III})_1 + (\mathrm{III})_2  ). 
}
By %
Inequality   \eqref{prop:Deltaineqc}  and by \cref{lem:estimate1}, respectively,
we find 
\begin{align}
\nn{(\mathrm{III})_{1}}  & \leq \frac{C|k|}{1-|\xi|}  \label{eq:Qtrick1}, \\
\nn{(\mathrm{III})_{2}} &  \leq \frac{C  F_{m,\xi}  }{ 1-|\xi|} 
 (|k|^{1/2} + |k|) .
 \label{eq:Qtrick2}
\end{align}
Now  inserting the above, we arrive at
\begin{align}
(\mathrm{II})_2 &= - \epsilon_\lambda(k) \cdot\nabla_\xi E_m(\xi)  \frac{e \rho(k)}{\sqrt{2 |k|}} \frac{1}{\omega_m(k) - k\cdot\nabla_\xi E_m(\xi)}  \psi_m(\xi) \nonumber \\ &\qquad - \epsilon_\lambda(k) \cdot \nabla_\xi E_m(\xi)  \frac{e \rho(k)}{\sqrt{2 |k|} } \frac{1}{\omega_m(k) - k\cdot\nabla_\xi E_m(\xi)} (  (\mathrm{III})_{1} +  (\mathrm{III})_{2} ). \label{eq:QQQQ1}
\end{align} 
Using  \eqref{eq:Qtrick1} and  \eqref{eq:Qtrick2}  we can estimate the second  term  on the right hand side of  \eqref{eq:QQQQ1} and find that it is of order $|k|^{-1}$. 
The claim now  follows by  collecting  estimates  and using 
that the first term on the right hand side of  \eqref{eq:QQQQ1} exactly cancels the second
 term on the right hand side of  \eqref{eq:conjann0}. 
\end{proof}

We still need an estimate  involving   derivatives. To this end,  we shall henceforth make  an explicit choice 
of the polarization vectors. 
After a possible unitary
transformation on Fock space we can always achieve that the
polarization vectors are given by
\begin{eqnarray} \label{eq:cheps}\label{def:eps}
\varepsilon_{1}(k)  = \frac{(k_2,-k_1, 0 ) }{\sqrt{k_1^2 + k_2^2}}
\quad \mathrm{and} \quad \varepsilon_{2}(k) = \frac{k}{|k|} \wedge
\varepsilon_1(k) \; .
\end{eqnarray}

\begin{prop} \label{lem:second2}
There exists a constant $C$ (depending on $\Lambda$)  such that the following holds. 
Let $e \in \R$ and $m > 0$, and suppose the energy Inequality \eqref{eq:eineq}
holds.  Let    $\xi \in \R^3$ with   $|\xi|  <  1$  such  that $\nabla E_m(\xi)$  exists. 
Then for almost all  $k$ with $|k| < \Lambda$,
\muu
{
\nn{ \nabla_k ( a_\lambda(k) U_m(\xi) \psi_m(\xi) )} \leq   \frac{C F_{m,\xi}^2    |e \rho(k)|   (  \widehat{D}_{m,\xi}(k)  + 1 )    }{(1-|\xi|)^2 |k|\sqrt{k_1^2 + k_2^2}   } .
}
\end{prop}
\begin{proof} 
First we note that for all nonzero $k \in \R^3$ and $\lambda = 1,2$, we have 
\begin{equation} \label{eq:derofepsi} 
\left| \frac{\partial}{\partial k_j}  \varepsilon_\lambda(k) \right| \leq \frac{1}{\sqrt{ k_1^2 + k_2^2} }  , \quad  j = 1 , 2,  3 . 
\end{equation} 
Moreover, as   $\rho$ is constant for $|k| < \Lambda$, we do not need to take derivatives of $\rho$ into account.  
Since $|k| < \Lambda$ it suffices to consider the leading order contributions as $|k|$ is small. 
From   differentiating  \eqref{eq:conjann0}   we find  
\begin{align}
& U_m(\xi)^* \nabla_k  a_\lambda(k) U_m(\xi) \psi_m(\xi)   =  \nabla_k  U_m(\xi)^* a_\lambda(k) U_m(\xi) \psi_m(\xi)  \label{eq:conjann}  \\
 & \quad  = \nabla_k \left(   a_\lambda(k) \psi_m(\xi)  +  \epsilon_\lambda(k)\cdot  \nabla_\xi E_m(\xi)  \frac{e \rho(k)}{\sqrt{2 \abs{k}}} \frac{1}{\omega_m(k) - k\nabla_\xi E_m(\xi)} \psi_m(\xi) \right) ,  \nonumber 
\end{align} 
where the derivative is understood in the strong sense. 
First we    calculate the derivative of the first term  using the  pull-through resolvent identity in \cref{lem:pullthrough}.
 
 For $k$ with $|k|<\Lambda$ we find by the product rule
\begin{align}
& \nabla_k a_\lambda(k) \psi_m(\xi)  \nonumber  \\
& \  = 
\nabla_k \left( -  \frac{e \rho(k)}{\sqrt{ 2 \abs k } }  R_{m,\xi}(k)\left( \epsilon_\lambda(k)\cdot  v(\xi) +  \sigma \cdot (k \wedge \epsilon_\lambda(k) \right) \psi_m(\xi)  \right) \nonumber \\ \label{eq:dif_est_1}
 &\  =   J_1 + J_2 ,  
\end{align} 
where we  introduced  
\begin{align*} 
 J_1  &:=  R_{m,\xi}(k) ( k-v(\xi)  + \nabla_k
\omega_m(k) ) R_{m,\xi}(k)\\
&\quad \times \frac{e \rho(k)}{\sqrt{2 | k|}}  \left( \epsilon_\lambda(k)\cdot  v(\xi) +  \sigma \cdot (k \wedge \epsilon_\lambda(k) \right) \psi_m(\xi)  %
\end{align*}
and 
$$
J_2  := J_{2,1} + J_{2,2}
$$
with
\begin{align*} 
J_{2,1}  &:=   -R_{m,\xi}(k)  \nabla_k \left( \frac{e \rho(k)}{ \sqrt{ 2 |k|}}     \epsilon_\lambda(k) \cdot v(\xi) \psi_m(\xi) \right)  ,   \\
J_{2,2}  & := 
-R_{m,\xi}(k)  \nabla_k \left( \frac{e \rho(k)}{\sqrt{ 2 | k| } }  S \cdot(k \wedge \epsilon_\lambda(k))  \psi_m(\xi) \right)  .
\end{align*}
To estimate  $J_{2,2}$ we use \eqref{prop:Deltaineqc} and    \eqref{eq:derofepsi}, which yields
\begin{align}
\label{eq:final J22 estimate}
\| J_{2,2}  \| \leq\frac{ C |e \rho(k)|}{1-|\xi|}  \frac{1}{|k|\sqrt{k_1^2 + k_2^2} } .
\end{align}
Next we consider  $J_{2,1}$ and  we write \begin{equation} \label{tricccck} v(\xi) = \nabla E_m(\xi) +  (v(\xi) - \nabla E_m(\xi)) \end{equation}  and  use 
\cref{lem:estimate1} \ref{lem:estimate111}  and  \eqref{eq:derofepsi}  to find  
\begin{align*}
\left\| J_{2,1}    + R_{m,\xi}(k)  \nabla_k \left( \frac{e \rho(k)}{ \sqrt{ 2 |k|}}     \epsilon_\lambda(k)\cdot  \nabla E_m(\xi) \psi_m(\xi) \right) \right\| 
\leq     \frac{C F_{m,\xi} |e \rho(k)|}{(1-|\xi|) |k|\sqrt{k_1^2+k_2^2}} .
\end{align*} 
Now we use   \eqref{eq:Rmxi_formula}--\eqref{eq:Qtrick2}  to replace the resolvent $R_{m,\xi}$  by $D_{m,\xi}$ and  get (noting that $\abs{ \nabla E_m(\xi) } < 1$)
\begin{align} 
\left\| J_{2,1}  + D_{m,\xi}(k)  \nabla_k \left(
\frac{e \rho(k)}{\sqrt{2 |k|} }  \epsilon_\lambda(k)\cdot   \nabla E_{m}(\xi) \psi_m(\xi) \right) \right\| \leq  \frac{C F_{m,\xi}  |e \rho(k)|  (  \widehat{D}_{m,\xi}(k)  + 1 )   }{(1-|\xi|) |k| \sqrt{k_1^2+k_2^2}} .\label{eq:finalderest0} 
\end{align} 
Finally let us consider   $J_1$. 
The term involving the spin has  a  lower order singularity  as well as  the term involving the $-k$. 
Thus we find  with  \eqref{prop:Deltaineqc} and \cref{lem:estimate1} \ref{en:estimate1_en2}
\begin{align} 
\left\| J_1 - R_{m,\xi}(k) (  \nabla_k \omega_m(k) - v(\xi)  ) R_{m,\xi}(k)\frac{e \rho(k)}{\sqrt{2 \abs k}}   \epsilon_\lambda(k)\cdot  v(\xi) \psi_m(\xi)  \right\| \leq 
\frac{C  F_{m,\xi} |e \rho(k)|}{(1-|\xi|)^2 |k|^{3/2}   }. \nonumber %
\end{align} 
Using again \eqref{tricccck} to replace the second $v(\xi)$ by $\nabla E_{m}(\xi)$, we find with \cref{lem:estimate1} \ref{lem:estimate111}  and \cref{lem:estimate1}  \ref{en:estimate1_en2}  
that 
\begin{align*} 
 \left\| J_1 - R_{m,\xi}(k) (  \nabla_k \omega_m(k) - v(\xi)  ) R_{m,\xi}(k)\frac{e \rho(k)}{\sqrt{2 \abs k}}   \epsilon_\lambda(k)\cdot   \nabla E_m(\xi) \psi_m(\xi)  \right\|  \leq  \frac{C F_{m,\xi}^2  |e \rho(k)|}{(1-|\xi|)^2 |k|^2   }.  
\end{align*} 
Now we use as before   \eqref{eq:Rmxi_formula}--\eqref{eq:Qtrick2}  to replace the  second resolvent $R_{m,\xi}$  by $D_{m,\xi}$ 
\begin{align} 
   &\left\| J_1 - R_{m,\xi}(k)  ( \nabla_k \omega_m(k) - v(\xi) ) 
 D_{m,\xi}(k) \epsilon_\lambda(k)\cdot  
\frac{e \rho(k)}{\sqrt{2 |k|} }\nabla E(\xi) \psi_m(\xi) \right\| \nonumber  \\ &\qquad \leq  \frac{C F_{m,\xi}^2  |e \rho(k)| (  \widehat{D}_{m,\xi}(k)  + 1 ) }{(1-|\xi|)^2  |k|^2   } . \label{eq:estofderest} 
\end{align} 
Repeating the above, that is, using  again first \eqref{tricccck}  to replace $v(\xi)$ by $\nabla E_m(\xi)$ and then   \eqref{eq:Rmxi_formula}--\eqref{eq:Qtrick2}
to replace the first resolvent $R_{m,\xi}$ by $D_{m,\xi}$, we arrive at 
\begin{align} 
 \left\| J_1 - ( \nabla_k \omega_m(k) -  \nabla E_m(\xi) ) 
[D_{m,\xi}(k)]^2 \epsilon_\lambda(k)\cdot  
\frac{e \rho(k)}{\sqrt{2 |k|} }\nabla E(\xi) \psi_m(\xi)  \right\|  \leq  \text{r.h.s. of } \eqref{eq:estofderest}  . \label{eq:finalderest1} 
\end{align}
Finally, note that the derivative of the second term in \eqref{eq:conjann} is
\begin{align}
 \nabla_k & \left(  \epsilon_\lambda(k)\cdot  \nabla_\xi E_m(\xi)  \frac{e \rho(k)}{\sqrt{2 \abs{k}}} D_{m,\xi}(k) \psi_m(\xi) \right) \nonumber  \\
&= D_{m,\xi}(k)  \nabla_k  \left(
\frac{e \rho(k)}{\sqrt{2 |k|} }  \epsilon_\lambda(k)\cdot   \nabla E_{m}(\xi) \psi_m(\xi) \right) \label{eq:final derivative1} \\
& \qquad - D_{m,\xi}(k)^2 ( \nabla_k \omega_m(k) - \nabla E_m(\xi))  \frac{e \rho(k)}{\sqrt{2 |k|} }  \epsilon_\lambda(k)\cdot   \nabla E_{m}(\xi) \psi_m(\xi).  \label{eq:final derivative2} 
\end{align} 
Thus, considering the decomposition \eqref{eq:dif_est_1} for the first term and combining \eqref{eq:finalderest0} with \eqref{eq:final derivative1}, noting \eqref{eq:final J22 estimate}, as well as \eqref{eq:finalderest1} with \eqref{eq:final derivative2} we obtain the desired upper bound. 
 
\end{proof}

\begin{lemma}[$y$-Bound] \label{thm:compact} 
  \label{lem:fourier2} Let $e \in \R$. 
 Suppose there exists an $m_0 > 0$ such that   \eqref{eq:eineq} holds for all $m \in (0,m_0)$.
  Let $ |\xi | < 1$. Then for every $M > 0$  there   exists  a constant $C$, and a $\delta
>0$ such that  for all $m \in (0,m_0)$ and all  $n \in \N$,
$$
\sum_{\lambda_1, \ldots, \lambda_n} \int   \sum_{i=1}^n n^{-1}
|y_i|^\delta \| F( U_m(\xi) \psi_m(\xi))_{(n)}(\lambda_1,{y}_1, \ldots ,
\lambda_n , y_n) \|^2 \, d y_1  \ldots  dy_n  \leq C \; ,
$$
whenever $| E_m(\xi) | < M$ and $\max \{| \partial_i^2 E_m(\xi) | : i=1,2,3 \} < M$.  
Here  $F( U_m(\xi) \psi_m(\xi))_{(n)}$ denotes the Fourier transform of
the $n$-photon component of $U_m(\xi)\psi_m(\xi)$.
\end{lemma}
\begin{proof} %
 We drop the subscript $m$.  We write  $\widehat{\psi} = U_m(\xi) \psi_m(\xi)$. Thus, $F \widehat{\psi}_{(n)}$ denotes the
Fourier transform of $\widehat{\psi}_{(n)}$ in all its $n$-components. We
define the functions
\begin{eqnarray*}
\widehat{\psi}_{(n)}(k) &:& (\lambda,  k_1, \lambda_1, \ldots , k_{n-1}, \lambda_{n-1} ) \mapsto \widehat{\psi}_{(n)}( k, \lambda ,
k_1, \lambda_1 ,  \ldots
, k_{n-1}, \lambda_{n-1}  ),  \\
F \widehat{\psi}_{(n)}(y) &:& (\lambda, y_1, \lambda_1,  \ldots , y_{n-1} , \lambda_{n-1}  ) \mapsto
F \widehat{\psi}_{(n)}(y , \lambda ,y_1, \lambda_1 , \ldots , y_{n-1} , \lambda_{n-1} ) \; .
\end{eqnarray*}

\vspace{0.5cm}

\noindent \underline{Step 1:} There exists a  $\delta >0$ and a
constant $C$ such that for all $a \in \R^3$,
$$
\int | 1 - e^{-\i ay} |^2 \| F\widehat{\psi}_{(n)}(y) \|^2 dy \leq
\left\{
\begin{array}{ll} C |a|^{\delta} \quad & \text{if} \  |a| < \frac{1}{2} \Lambda , \\
C \quad & \text{if} \  |a| \geq \frac{1}{2}  \Lambda  .  \end{array}
 \right.
$$
The claim follows easily for $|a|\geq \frac{1}{2} \Lambda$, since $\widehat{\psi}$ is a
normalized state in Fock space and $|1-e^{-\i ay}| \leq 2$. Now lets
consider the case $|a| < \frac{1}{2}  \Lambda$.  By the Fourier transform, we have the
identity
\begin{align}
&  \int | 1 - e^{-\i ay} |^2 \| F\widehat{\psi}_{(n)}(y) \|^2 dy  \nonumber \\
& \quad = \int \|
\widehat{\psi}_{(n)}(k + a ) - \widehat{\psi}_n(k) \|^2 dk \nonumber \\ 
& \quad  =  \int_{|k| < \Lambda - |a|}   \|
\widehat{\psi}_{(n)}(k + a ) - \widehat{\psi}_n(k) \|^2 dk   +   \int_{\Lambda - |a| \leq |k| }   \|
\widehat{\psi}_{(n)}(k + a ) - \widehat{\psi}_n(k) \|^2 dk .  \label{eq:yboundfourier}
\end{align}
To estimate the second integral  we use \cref{prop:a_est_2}  and observe that the integrand vanishes for $\abs k > \Lambda + \abs a$,
\begin{align}
\int_{\Lambda - \abs a \leq \abs k }  \|
\widehat{\psi}_{(n)}(k + a ) - \widehat{\psi}_n(k) \|^2 dk  &\leq   \text{const.}  \int_{\Lambda - \abs a \leq \abs k \leq \Lambda + \abs a}  \left( \frac{1}{\abs{k+a}^2} + \frac{1}{\abs{k}^2} \right)dk \nonumber \\
 &\leq   \text{const.}   \int_{\Lambda - 2 \abs a \leq \abs k \leq \Lambda + 2 \abs a} \frac{1}{\abs{k}^2} dk \nonumber \\
 &\leq   \text{const.}   \abs a , \label{eq:intofpsi1} 
\end{align}
where \text{const.} denotes a numerical constant changing from line to line. 
Next we estimate the first integral  and 
assume  $\abs k < \Lambda - \abs a$.  Using \cref{lem:second2} we find 
\begin{align}
\| \widehat{\psi}_{(n)}( k + a ) - \widehat{\psi}_{(n)}(k) \|  & =  \left\| \int_0^1
\left(
\frac{d}{dt} \widehat{\psi}_{(n)}(k + t a) \right) dt  \right\|  \nonumber \\
& \leq |a| \int_0^1 \| \nabla_k \widehat{\psi}_{(n)}(k + t a) \| dt  \nonumber \\
& \leq  \text{const.}   |a| \int_0^1 \frac{ \rho(k + t a )}{ |
k + t a | | \pi_3( k + t a ) | } dt \; , \label{eq:intofpsi} 
\end{align}
where $\pi_3$ denotes the projection in $\R^3$ along the 3-axis and
const. denotes a finite constant independent of $n$. Let $\pi_a$
denote the projection in $\R^3$ along the vector $a$ and let
$\pi_{3,a}$ denote the projection in the $(1,2)$-plane along $\pi_3
a$  (with convention that $\pi_{3,a} = \pi_3$,  if $\pi_3 a = 0$). We find  from \eqref{eq:intofpsi}
\begin{eqnarray}  \label{eq:ybound1}
\| \widehat{\psi}_{(n)}(k + a ) - \widehat{\psi}_{(n)}(k) \| \leq \text{const.}  \frac{
|a| }{|\pi_a(k) | | \pi_{3,a}(k) |} \; .
\end{eqnarray}
On the other hand using \cref{prop:a_est_2} we obtain 
\begin{eqnarray}  \label{eq:ybound2}
\| \widehat{\psi}_{(n)}(k + a) - \widehat{\psi}_{(n)}(k) \| \leq \text{const.}   \left(
\frac{\rho(k+ a)}{|k+a| }  + \frac{\rho(k)}{|k|}
\right) \; .
\end{eqnarray}
Introducing Inequalities \eqref{eq:ybound1} and \eqref{eq:ybound2}
into the second integral of  \eqref{eq:yboundfourier}, we find for any $\theta$
with $0 \leq \theta \leq 1$,
\begin{eqnarray*}
 \lefteqn{ \int_{ 
\abs k < \Lambda - \abs a} \| \widehat{\psi}_{(n)}(k + a ) - \widehat{\psi}_n(k) \|^2 dk  } \\
&&= \text{const.}  |a|^{2 \theta} \int_{|k| < \Lambda - |a|}  \frac{1}{|\pi_a(k) |^{2\theta}
|\pi_{3,a}(k) |^{2 \theta} } \left( \frac{\rho(k+a)}{|k+a|}
+ \frac{\rho(k)}{|k|} \right)^{2(1-\theta)} \, dk.
\end{eqnarray*}
Now we use Young's inequality: $b c \leq b^p/p + c^q/q$, whenever
$p,q
> 1$ and $p^{-1} + q^{-1} = 1$; and the convexity of $x \mapsto x^{2(1-\theta)q}$ on $\R_+$, for
$0 < \theta < 1/2$. Thus for   $0 < \theta < 1/2$,
\begin{align}
&  \int_{  \abs k < \Lambda - \abs a} \| \widehat{\psi}_{(n)}(k + a ) - \widehat{\psi}_n(k) \|^2 dk 
\nonumber \\  
& \quad \leq |a|^{2 \theta} \text{const.}   \int_{|k| \leq \Lambda  }
\Bigg( \frac{1}{|\pi_a(k)|^{4 \theta p}} +
\frac{1}{|\pi_{3,a}(k)|^{4 \theta p}}   + \left[ \frac{1}{|k+a|}
\right]^{2 (1-\theta) q}+ \left[ \frac{1}{|k|}
\right]^{2 (1-\theta) q} \Bigg) dk \; . \label{eq:intofpsi2} 
\end{align}
For any $q$ with $1<q \leq 3/2$, we can choose $\theta > 0$ sufficiently
small such that the right hand side is finite. 
Inserting  \eqref{eq:intofpsi1}  und  \eqref{eq:intofpsi2}  into \eqref{eq:yboundfourier}
we obtain the desired estimate.

 \vspace{0.5cm}

\noindent \underline{Step 2:} Step 1 implies the statement of the
Lemma.

From Step 1 we know that there exists a finite constant $C$ such
that
$$
\int \frac{|1 - e^{-\i ay} |^2 \| F\widehat{\psi}_{(n)}(y)
\|^2}{|a|^{\delta/2} } dy \frac{da}{|a|^3} \leq C   \; .
$$
After interchanging the order of integration  and a change of
integration variables  $b = |y| a $, we find
\begin{align*}
C  &\geq \int \| F\widehat{\psi}_{(n)}(y) \|^2 \int \frac{ | 1 -
e^{-\i ay} |^2}{|a|^{\delta/2} } \frac{da}{|a|^3} dy \\
&= \int \|
F\widehat{\psi}_{(n)}(y) \|^2 |y|^{\delta/2} \underbrace{ \int \frac{
| 1 - e^{-\i by/|y|} |^2}{|b|^{\delta/2} } \frac{db}{|b|^3}  }_{=: \ c}  dy  ,
\end{align*}
where $c$ is nonzero and does not depend on $y$.
\end{proof}

\section{Compactness Argument} 
\label{sec:maincompproof}

Using a compactness argument  we show Theorem 
\ref{thm:main222}. It states  that for suitable $\xi$ there exists a   sequence $(m_j)$ of  positive  
numbers converging to zero such that 
 $U_{m_j}(\xi)\psi_{m_j}(\xi)$ converges to a vector in the Hilbert space, say $\widehat{\psi}_0$.
Then we prove the main \cref{thm:main result},
which states that this vector  $\widehat{\psi}_0$  is indeed the ground state of 
the renormalized Hamiltonian using a semi-continuity argument of quadratic forms. 

\begin{thm}   \label{thm:mainmain} 
 Let  $e \in \R$, and    suppose there exists an $m_0 > 0$ such that the energy Inequality 
  \eqref{eq:eineq} holds for all $m \in (0,m_0)$. 
Let  $\xi \in \IR^3$  with $\abs {\xi} <  1$ such that $E$ is differentiable at $\xi$. 
Suppose there  exists  a sequence 
$(m_j)_{j \in \N}$ of positive numbers converging to zero such that
\begin{itemize}
\item[(i)]    $E_{m_j}$ is differentiable at $\xi$,
\item[(ii)]  $\nabla E_{m_j}(\xi) \overset{j \to \infty}\longrightarrow \nabla E(\xi)$,
\item[(iii)]   for every $l=1,2,3$ the  partial derivatives  $\partial_l^2 E_{m_j}(\xi)$ exist and  satisfy  $$\sup_j (-  \partial_l^2 E_{m_j}(\xi)) < \infty .$$ 
\end{itemize} 
 Then $E(\xi)$ is an eigenvalue of $\widehat{H}(\xi)$, and there exists  a subsequence of 
 $(U_{m_j}(\xi) \psi_{m_j}(\xi))_{j \in \N}$ converging to the eigenvector.  
\end{thm} 

\begin{proof}

\vspace{0.5cm}  \noindent \underline{Step 1:} The sequence of vectors  $\widehat{\psi}_{m_j}(\xi) := U_{m_j}(\xi) \psi_{m_j}(\xi)$, $j \in \N$,  lies in a compact subspace  of  the reduced Hilbert space $\HH=\C^{2s+1} \otimes \FF $. 

\vspace{0.5cm}

Let $L$ be  the self-adjoint operator  associated to
the nonnegative and closed quadratic form $q$   in $\HH$ 
defined by
\[
q(\phi)  :=  \inn{ \phi , \dG(\Id)  \phi }  + \sum_{n=1}^\infty n^{-3} \inn{ 
\widehat{\phi}_{(n)} , \sum_{i=1}^n |y_i|^{\delta}
\widehat{\phi}_{(n)}  } + \inn{ \phi , \Hf   \phi }  
\]
on the natural form domain $D(q)$. We choose $\delta
> 0$ such that Lemma \ref{thm:compact} holds. By this and \cref{prop:a_est_2,prop:nabla_E_converging}, there exists a finite $C$ such that for all $m$
with $0 < m <m_0$,
\[
\psi_m \in \mathcal{K} := \{ \phi \in D( q) : \| \phi \|  \leq 1,  q(\phi)
\leq C \} \; . 
\]
 The set $\mathcal{K}$ is a compact subset of $\HH$, provided $L$ has compact resolvent 
\cite[Theorem XIII.64]{ressim:ana}. Hence it remains to show that $L$ has compact
resolvent. The operator $L$ preserves the $n$-photon sectors. Let
$L_n$ denote the restriction of $L$ to the $n$-photon sector. From
Rellich's criterion  \cite[Theorem XIII.65]{ressim:ana}   it follows that $L_n$ has  compact resolvent.
Therefore $\mu_l(L_n) \to \infty$ as $l$ tends to infinity, where
$\mu_l$ denotes the $l$-th eigenvalue obtained by the min-max
principle. Moreover, since $\mu_l(L_n) \geq n$ for all $l, n$, it
follows that $\mu_l(L) \to \infty$ as $l \to \infty$. Hence $L$ has
a compact resolvent.

\vspace{0.5cm}

\vspace{0.5cm}  \noindent \underline{Step 2:} 
The sequence in Step 1 has a subsequence which converges to a normalized vector $\widehat{\psi}_0 \in \HH$. 
\vspace{0.5cm}

This follows directly from Step 1 and the property of compact sets.

\vspace{0.5cm}

\noindent \underline{Step 3:} The vector  $\widehat{\psi}_0$ is an eigenvector   of the renormalized 
fiber Hamiltonian $\widehat{H}(\xi)$ with eigenvalue $E(\xi)$.

\vspace{0.5cm}

Using lower semicontinuity of nonnegative quadratic forms
it  follows from Step 2  and Proposition  \ref{convofenergyexp} 
that for almost all $\xi$ with $|\xi|< 1$ 
\begin{align*}
 0  & \leq \inn{  \widehat{\psi}_0(\xi) , ( \widehat{H}(\xi)  - E(\xi) ) \widehat{\psi}_0(\xi) }  \\
& \qquad \leq \liminf_{i \to \infty} \inn{ U_{m_j}(\xi)  \psi_{m_j}(\xi) , ( \widehat{H}(\xi) - E(\xi) )  U_{m_j}(\xi) \psi_{m_j} } = 0 , 
\end{align*}
i.e., that $\widehat{\psi}_0(\xi)$ is a ground state of $\widehat{H}(\xi)$. 
\end{proof}

Now the above theorem implies together with   Proposition \ref{prop:nabla_E_converging} 
\cref{thm:main result,thm:main222}.

\begin{proof}[Proof of \cref{thm:main result,thm:main222}] 
  Let $e \in \R$. By Theorem   \ref{thm:nospingroundstateenineq}   in case $s=0$ or  by assumption in case $s=1/2$,
there exists an $m_0 > 0$, such that   \eqref{eq:eineq} holds for all $m \in (0,m_0)$. 
Let  $(m_j)_{j \in \N}$  be any  sequence in $(0,m_0)$  which converges to zero. 
 Then by Proposition \ref{prop:nabla_E_converging}  there exists a set $D \subset \R^3$ of full Lebesgue measure  with 
the following property: For all $\xi \in D$ 
the functions  $E_{m_j}$, $j \in \N$,  and $E$ are differerentiable and 
\begin{itemize}
\item[(a)]  $\nabla E_{m_j}(\xi) \overset{j \to \infty}\longrightarrow \nabla E(\xi)$,
\item[(b)]   the second partial derivatives  $\partial_l^2 E_{m_j}(\xi)$ exist and  satisfy for every $l=1,2,3$ that $\liminf_j (-  \partial_l^2 E_{m_j}(\xi)) < \infty$ . 
\end{itemize} 
Now pick any $\xi \in D$ with $|\xi| < 1$ and fix it.
Then the 
assumptions of  \cref{thm:mainmain}  hold. Therefore,  \cref{thm:main result} and \cref{thm:main222}
 follow directly from  \cref{thm:mainmain}.
\end{proof}

\newcommand{\h}{\mathfrak{h}}

\appendix

\section{Energy Estimates} 

\label{secexprea0} 

In this section we collect a few well known properties related to the canonical 
commutation relations, which we need in particular in Section \ref{sec:trafham}. 
\begin{lemma} \label{lem:canest} 
Let $f_1,\ldots,f_n \in L_{(n)}^2(\R^3 \times \Z_2)$. Then for any $\psi \in \FF$ we have 
\begin{align*}
\| a(f_1) \cdots a(f_n) \psi \| &  \leq  \left( \prod_{j=1}^n \|  f_j \omega_m^{-1/2} \| \right) \| \Hfm^{n/2} \psi \|   .
\end{align*} 
\end{lemma} 
\begin{proof}
In the following we use the notation $\underline{k} = (\lambda,k)$ and $\int (\cdots)  d \underline{k} = \sum_{\lambda=1,2} \int ( \cdots ) dk$.
By the definition of the annihilation operator, by Cauchy-Schwarz and Fubini, we find
\begin{align*}
& \| a(f_1) \cdots a(f_n) \psi \| \\ & \leq \int  | f_1(\uk_1) \cdots f_n(\uk_n) | \| a(\uk_1) \cdots a(\uk_n) \psi \| d\uk_1 \cdots d\uk_n  \\ 
 & \leq      \left( \prod_{j=1}^n   \|  f_j \omega_m^{-1/2} \| \right) \left( \int   \omega_1(\uk_1) \cdots \omega_n(\uk_n)   \| a(\uk_1) \cdots a(\uk_n) \psi \|^2  d\uk_1 \cdots d \uk_n  \right)^{1/2}  .
\end{align*}  
To estimate the second factor we use 
\begin{align*}
&   \int   \omega_1(\uk_1) \cdots \omega_n(\uk_n)   \| a(\uk_1) \cdots a(\uk_n) \psi \|^2  d\uk_1 \cdots d\uk_n    \\ 
& =  \int   \omega_1(\uk_2) \cdots \omega_n(\uk_n)   \| \Hfm^{1/2}   a(\uk_2) \cdots a(\uk_n) \psi \|^2  d\uk_2 \cdots d\uk_n \\
& =  \int   \omega_1(\uk_2) \cdots \omega_n(\uk_n)   \| \Hfm^{1/2}   a(\uk_2) \cdots a(\uk_n)  \Hfm^{-1/2} \Hf^{1/2} \psi \|^2  d\uk_2 \cdots d\uk_n \\
& =  \int   \omega_1(\uk_2) \cdots \omega_n(\uk_n)   \| \Hfm^{1/2} (   \Hfm+\sum_{j=2}^n \omega_m(k_j))^{-1/2}   a(\uk_2) \cdots a(\uk_n) \Hfm^{1/2} \psi \|^2  d\uk_2 \cdots d\uk_n \\
& \leq   \int   \omega_1(\uk_2) \cdots \omega_n(\uk_n)   \|  a(\uk_2) \cdots a(\uk_n) \Hfm^{1/2} \psi \|^2  d\uk_2 \cdots d\uk_n \\
& \vdots \\
& \leq \| \Hfm^{n/2} \psi \|^2 ,
\end{align*}  
where we used the pull-through formula, cf. \cite[Lemma A.1]{BachFroehlichSigal.1998a}, and that $\Hf \varphi = 0$ implies $a(f) \varphi = 0$.
 
\end{proof} 

\begin{lemma} \label{lem:basichfest} For any $f \in L^2_{(m)}(\Z_2 \times \R^3)$ and $\psi \in \FF$ 
\begin{align*} 
\| \phi( f ) \psi \| &  \leq   \sqrt{2} \| f \|_{(m)}  \| ( \Hfm + 1 )^{1/2} \psi \| , \\
\| \phi( f_1 ) \phi(f_2)  \psi \| &  \leq 2  \| f_1 \|_{(m)}    \| f_2 \|_{(m)}  \| ( \Hfm + 1 ) \psi \|   . 
\end{align*}  
\end{lemma} 
Note that the proof of \cref{lem:basichfest} can be found in \cite[Theorem 5.18]{Arai.2018}, but for the convenience of the reader we give a proof below.
\begin{proof}
 
From the canonical commutation relations we get
\begin{align*} 
\| a(f)^* \psi \|^2  &   = \inn{ a^*(f) \psi, a^*(f) \psi } = \| f \|^2 \| \psi \|^2  + \| a(f) \psi \|^2  . 
\end{align*} 
Now the first identity follows from the triangle inequality and Lemma   \ref{lem:canest}.  
 
For the second identity we use again the canonical commutation relations and find 
\begin{align*} 
& \| a^*(f_1) a^*(f_2)  \psi \|^2 \\
 & \quad    = \inn{ a^*(f_1) a^*(f_2)  \psi, a^*(f_1) a^*(f_2)  \psi } \\
& \quad \leq  \| a(f_1) a(f_2) \psi \|^2 + 2  \| f_1 \|^2  \| a(f_2) \psi \|^2  +   2   \| f_2 \|^2  \| a(f_1) \psi \|^2   +  2  \| f_1\|^2 \| f_2 \|^2   , 
\end{align*} 
and 
\begin{align*} 
\| a^*(g_1) a(g_2)  \psi \|^2  &   = \inn{ a^*(g_1) a(g_2)  \psi, a^*(g_1) a(g_2)  \psi }   = \| a(g_1) a(g_2) \psi \|^2 +  \| g_1 \|^2  \| a(g_2) \psi \|^2   \\
\| a(f_1) a^*(f_2)  \psi \|^2  &   \leq  ( \| f_1 \| \| f_2 \|  \| \psi \| + \| a^*(f_2) a(f_1)  \psi \| )^2   . 
\end{align*} 
The second inequality follows now by collecting estimates,  the triangle inequality and Lemma   \ref{lem:canest}.
 
\end{proof} 

\section{Statements about  CCR Algebras} 

\label{secexprea}

\begin{lemma} \label{lem:comdgamma} Let $A$ be a self-ajoint operator in $\hh$ and $f \in D(A) \subset \hh$. 
Then we have the relations 
\begin{align*} 
\dG(A) a^*(f) & = a^*(f) \dG(A) + a^*(A f), \\
\dG(A) a(f) & = a(f) \dG(A) - a(A f), \\
\dG(A) \phi(f)  & = \phi(f) \dG(A) - \i \phi(\i A f ) 
\end{align*}
on $\FF_{\fin}(D(A))$. 
\end{lemma} 
\begin{proof}
The first identity follows directly from the definition of the creation operators and  $\dG(A)$. The second 
identity follows from the first by taking adjoints. The last identity follows from the first two. For details see for example \cite[Proposition 5.10]{Arai.2018}. 
\end{proof} 

\begin{lemma} \label{lem:comm}  Let  $f , g \in \hs$. Then for every $\psi \in D( \dG(\Id))$  one has 
$$
( \phi(f) \phi(g) - \phi(g) \phi(f) ) \psi  = \i \Im\langle f , g \rangle  \psi  .
$$
\end{lemma} 
\begin{proof}
This follows directly from the canonical commutation relations  \eqref{eq:ccr}. For details see for example \cite[Proposition 5.14]{Arai.2018}. 
\end{proof} 

For $f \in \hh$  we define   $$W(f) := \exp( \i \phi(f)) = \exp( \i \pi(\i f ))  . $$  
Recall that $\FF_{\fin}(\hs)$ denotes the subspace of elements $\psi \in \FF(\hs)$ such that $\psi_n = 0$ for all but finitely many $n$. 
For the definition of an analytic vector for an operator we refer the reader to \cite{ressim:fou}. 
\begin{lemma}\label{lem:anaop} 
 Let $f \in \hh$ and $\psi \in \FF_{\fin}(\hs)$. Then for all $t \in \C$ we have 
$$	
\sum_{n=0}^\infty  \frac{ \| \phi(f)^n  \psi \|}{n!}|t|^n < \infty  . 
$$
In particular, $\FF_{\fin}(\hs)$ is a dense set of analytic vectors for $\phi(f)$. 
\end{lemma} 
\begin{proof} 
See the proof of \cite[Proposition 5.2.3]{bratellirobinsion2}.
\end{proof} 

\begin{lemma} \label{lem:bratt}  Let $f , g \in \hh$. 
\begin{enumerate}[(a)] 
\item Then $W(f) D(\phi(g)) = D(\phi(g))$ and 
\[
W(f) \phi(g)  W(f)^* = \phi(g) - \Im \langle  f , g \rangle .
\]
\item Then $W(f) D(a^\#(g)) = D(a^\#(g))$ and 
\begin{align*}
W(f)  a(g) W(f)^*  & = a(g)  - \i  2^{-1/2}  \langle{ g ,f }\rangle   , \\  
W(f)  a^*(g) W(f)^* & = a^*(g) + \i   2^{-1/2}  \langle{ f ,g }\rangle .
\end{align*} 
\end{enumerate} 
\end{lemma} 
The proof of the lemma before can be found for example in \cite[Proposition 5.2.4]{bratellirobinsion2} and \cite[Corollary 5.12]{Arai.2018}. For the convenience of the reader we sketch a proof below.
\begin{proof}[Sketch of proof]
(a)  By Lemma \ref{lem:anaop} every   $\psi \in \FF_{\fin}(\hs)$ is analytic for $\phi(f)$.
Thus one can define $\phi(g) W(f)^*$ 
on $\psi$ by a power series expansion, which  yields the identity 
\[
\phi(g) W(f)^* \psi = W(f)^* \phi(g) - \Im \langle f , g \rangle \psi .
\]
Since $\FF_{\fin}(\hs)$ is an operator core for $\phi(g)$,  the claim now  follows since $\phi(g)$ 
is by definition a closed operator. For details we refer the reader to \cite[Proposition 5.2.4]{bratellirobinsion2}. 

Part (b) is shown similarly. 
\end{proof} 
\begin{lemma} \label{lem:dgammaw}    Let $A$ be a self-adjoint operator in $\hh$ and $f \in D(A)$. Then we have  $W(f) D(\dG(A)) \supset D(\dG(A)) \cap D(\phi(\i A f))$ and  
as an identity on the latter 
\begin{equation} \label{eq:wdgammarel} 
W(f)  \dG(A)  W(f)^* = \dG(A) - \phi( \i A f ) + \frac{1}{2}  \langle  f , A f  \rangle  .
\end{equation} 
\end{lemma} 
\begin{proof} 
First we note that by definition of $\dG(A)$ we have on $\FF_{\fin}(D(A))$
$$
e^{ \i \dG(A) t } \phi(f)^n = \phi( e^{ \i A t } f )^n e^{  \i \dG(A) t } .
$$
On the same domain we can  differentiate with respect to $t$,
and find at $t = 0$  
\begin{align*}
\i \dG(A) \phi(f)^n  & =  \phi(f)^n \i \dG(A) + 
\sum_{l=0}^{n-1} \phi( f)^l \phi(\i A f ) \phi(f)^{n-1-l} \\ 
  & =  \phi(f)^n \i \dG(A) + n \phi(f)^{n-1} \phi( \i A f )  - \i  \frac{n ( n - 1)}{2} \phi^{n-2}(f)  \langle f , A f \rangle  ,
\end{align*} 
where we used Lemma \ref{lem:comm} for the last identity. Multiplying with $(-\i)^n(n!)^{-1}$ and summation over 
$n \in \N_0$ yields on  $\FF_{\fin}(D(A))  $
$$
\i \dG(A) W(f)^* = \i W(f)^* \dG(A) - \i W(f)^* \phi(\i A f ) +  \i \frac{1}{2} W(f)^* \langle f , A f \rangle . 
$$
Thus, it follows that   \eqref{eq:wdgammarel}  holds  on  $\FF_{\fin}(D(A)) $.
If $\psi \in D(\Gamma(A)) \cap D(\phi(\i A f))$, then $\psi_n = 1_{N \leq n} \psi \in  
\FF_{\fin}(D(A)) $. Now $(\dG(A) -  \phi(\i A f ) + \frac{1}{2} \langle f , A f \rangle ) \psi_n \to ( \dG(A) -  \phi(\i A f ) + \frac{1}{2} \langle f , A f \rangle ) \psi $. Therefore, $\psi \in D(W(f) \dG(A) W(f)^*) = W(f) D( \dG(A))$ and by closedness \eqref{eq:wdgammarel} holds for $\psi$. 
\end{proof}

\section{Some Statements  about  Convex Functions}

\label{sec:B} 

\begin{prop} \label{prop:convexlms} Let $F \: \R^n \to \R$ be a function that satisfies the following conditions: 
\begin{itemize}
\item[(i)]  $F(0) \leq F(x)$, 
\item[(ii)]  $F(x) \leq \frac{x^2}{2} + F(0) $, 
\item[(iii)]  $x \mapsto  \frac{x^2}{2} - F(x)$ is a convex function. 
\end{itemize} 
Then %
we have 
$$
F(x - k) - F(x) \geq \left\{ \begin{array}{ll} -|k||x| + k^2/2 & \text{ , if } |k| \leq |x|, \\ - x^2/2  & \text{ , if } |k| \geq |x| . \end{array} \right.   
$$
\end{prop} 
\begin{proof}  A proof  %
is given in \cite[Appendix A]{LMS06}. 
\end{proof} 
 
\begin{lemma}\label{lem:eleconv} Let $g \: (a,b)  \to \IR$ be convex. Then for 
any compact interval $[c,d] \subset (a,b)$, the function $g$ is Lipschitz continuous on $[c,d]$ with Lipschitz constant $K$ bounded  by  $  \max\{ \epsilon^{-1} |g(c)-g(c-\epsilon)|, \epsilon^{-1}  |g(d+\epsilon)-g(\epsilon)|\}$ for any $\epsilon > 0$ such that $[c-\epsilon, d + \epsilon ] \subset (a,b)$.
\end{lemma}
\begin{proof}  If $F$ is convex, then for all $s,t,s',t' \in (a,b)$ such that $s \leq s' < t'$ and $s < t \leq t'$,
$$
\frac{F(t)-F(s)}{t-s} \leq \frac{ F(t') - F(s')}{t'-s'} ,
$$
see for example \cite{folland1999real}. The claim now follows from the above inequality. 
\end{proof}

\begin{lemma} \label{estonderconv}  Let $f \: \R \to \R$ be a convex function with $f(x) = f(-x)$ for all $x \in \R$.
Then $f$ has a global minimum in $0$ and in points $x \geq 0$ where $f$ is differentiable 
we have $f'(x) \geq 0$. 
\end{lemma} 
\begin{proof}
By convexity and symmetry  $f(0) = f(\frac{1}{2}{x} - \frac{1}{2} x ) \leq \frac{1}{2} ( f(x) + f(-x) ) = f(x)$. 
If $f$ is differentiable at $x > 0$,  then by elementary convexity we conclude that 
$0  \leq \frac{f(x)-f(0)}{x-0}   \leq f'(x)$. The case when $x=0$ follows from the minmality property.
\end{proof}

Let us state the following theorem, which implies in view of the previous lemma 
that every convex function is almost everywhere differentiable. 

\begin{thm}[Rademacher] \label{thm:rade}  Let $U \subset \IR^n$ be open and  $f \: U \to \IR$ be Lipschitz continuous.
Then  $f$ is almost everywhere in $U$ differentiable.
\end{thm} 
A proof can be found in  \cite{evans}.  
In fact, for convex functions the second derivative exists almost everywhere. This is the statement of 
the so called Alexandrov theorem, \cite{Alex39}.  

\begin{thm}[Alexandrov] \label{alexthm} Let $U \subset \IR^n$ be open and $f : U \to \IR$ convex. Then $f$ 
has a second derivative almost everywhere. 
\end{thm} 

In the following two lemmas we study sequences of convex functions which converge pointwise. This 
will be needed to contol properties of the ground state energy as the mass regularization  is removed.

\begin{lemma}
\label{lemma:derivative_converging}
Let $f, (f_n)_{n \in \IN}$ be convex functions from $\IR$ to $\IR$ with $f_n \ton f$ pointwise. Then then there exists a set $D$ with $\R \setminus D$ a Lebesgue 
null set such that on $D$  the functions $f$ and  $f_n$, $n \in \N$,   are differentiable 
 and pointwise $ f_n' \to  f'$.
\end{lemma}
\begin{proof}
By \cref{lem:eleconv} we see  that each convex function is locally Lipschitz continuous,  hence 
almost everywhere differentiable by Rademacher's theorem, Theorem \ref{thm:rade}. 
 Hence there exists  a set $D$, where $f_n$, $n \in \IN$, and $f$ are differentiable such that $\IR^d \setminus D$ is a set of Lebesgue measure zero.
 By pointwise convergence and  again  \cref{lem:eleconv} we see that $f_n$ is a family of uniformly equicontinuous functions on any compact interval. Now uniform equicontinuity and pointwise convergences imply uniform convergence
 (see for example \cite[Theorem I.27]{rs1}).   So let  $x \in D$ be arbitary. Then we have that  $\delta_n := \sup_{y \in [x-1,x+1]} 
|f(y)-f_n(y)| $ satisfies  $\delta_n \ton 0$.
Let  $n_0 \in \IN$ be  such that $\delta_n \leq 1$ for all $n \geq n_0$.
 Since $f_n$ is convex and differentiable in $x \in D$, we have
\muu
{
f_n'(x) \leq \frac{f_n(x+\sqrt \delta_n) - f_n(x)}{\sqrt \delta_n}.
}
Therefore, for $n \geq n_0$, we get
\muu
{
 f_n'(x) \leq \frac{f(x+\sqrt \delta_n) - f(x) + 2\delta_n}{\sqrt \delta_n}
\ton  f'(x).
}
This implies $\lim\sup_{n\to\infty} f_n'(x) \leq f'(x)$.
Analogously,  we obtain  $\lim\inf_{n\to\infty} f_n'(x) \geq  f'(x)$. This shows that $ f_n'(x) \ton  f'(x)$.
\end{proof}

\begin{lemma}
\label{lemma:onvex2david}
Let $f $ and $f_n$, $n\in \N$,  be convex functions from $\IR$ to $\IR$ with $f_n \ton f$ pointwise. Then there exists a set $D \subset \IR$ such that $\IR \setminus D$ has Lebesgue
measure zero and the following holds.
\begin{itemize}
\item[(a)] For all $x \in D$ and $n \in \IN$ the function $f_n$ is twice differentiable in $x$ and $f_n''(x) \geq 0$.
\item[(b)] For all $x \in D$ we have $\liminf_{n} f_n''(x) < \infty$.
\end{itemize}
\end{lemma}
\begin{remark}
We note that  (a) follows directly from Alexandrov's theorem, \cref{alexthm}.  However, in the  proof of (b), which we will present below, 
Part   (a) will follow as  an intermediate step. 
\end{remark}
\begin{proof} 
Assume that $g$ is a convex function. Then it is locally Lipschitz continuous by \cref{lem:eleconv}.
Thus it is absolutely continuous and therefore by the fundamental theorem of calculus for Lebesgue integrals (see for example \cite{folland1999real} Theorem 3.35),
$g$ is almost everywhere differentiable, $g'$ is  locally in $L^1$, and for every $x \in \R$ we have
\[
g(x) = \int_0^x g'(\xi) d \xi .
\]
By convexity  $g'$ is  monotone on the set of points where $g$ is differentiable.  By possibly extending $g'$ to a monotone function on $\R$,
we can assume without loss that  $g'$  is a monotone function on $\R$.
Since monotone functions have at most countably many discontinuities (see for example \cite{folland1999real} Theorem 3.23), we can assume furthermore  that $g'$ is right continuous.
Thus, there exists a unique measure $\mu$ on $\R$ such that for all $a,b \in \R$
\[
\mu((a,b]) = g'(b)-g'(a)
\]
(see for example \cite[Theorems 1.16 and 1.18]{folland1999real}). Let $\mu = h d \lambda + \rho$ be its Lebesgue-Radon-Nykodim representation, where $\rho$ is mutually singular to the Lebesgue measure $\lambda$ and
 $h  \geq 0$ is Borel  measurable (see for example \cite[Theorem 3.8]{folland1999real}).  Then  by taking suitable families shrinking nicely to $x$
 we see that (see for example \cite[Theorem 3.22]{folland1999real})  $g'$ is almost everywhere differentiable and
 $g'' = h$. This shows (a). Furthermore,  we find
 \begin{equation} \label{eq:diffder} 
 g'(b)-g'(a) = \mu((a,b])  =  \int_a^b  g''(x) d x  + \rho((a,b]) \geq \int_a^b  g''(x) d x  .
 \end{equation} 
By \cref{lemma:derivative_converging} there exist for each $m \in \N$ two numbers  $a \leq -m ,  m \leq b$ for which   $f_n'(b)$ and $f_n'(a)$ converge as $n \to \infty$. Thus, 
inserting $f_n$ into \eqref{eq:diffder}, and using the Lemma of Fatou, we obtain 
$$
\lim_n ( f_n'(b) - f_n'(a))      \geq \int_a^b \liminf_n f_n''(x)dx . 
$$ 
It follows that  $\liminf_n f_n''(x) < \infty$ almost everywhere.  This shows (b). 
\end{proof}

\renewcommand*{\bibfont}{\footnotesize}
\printbibliography

\end{document}